\documentclass{article}
\usepackage{amsthm, amsmath, amssymb, amsbsy, amsfonts, latexsym, euscript,amscd}
\usepackage{calrsfs}
\usepackage{graphicx}
\usepackage{newa4}

\def\phi{\varphi}
\def\hat{\widehat}
\def\tilde{\widetilde}
\def\bar{\overline}
\def\cal{\mathcal}
\def\kappa{\varkappa}
\def\frak{\mathfrak}

\newtheorem{theorem}{\bf \indent Theorem}[section]
\newtheorem{proposition}{\bf \indent Proposition}[section]
\newtheorem{lemma}{\bf \indent Lemma}[section]
\newtheorem{corollary}{\bf \indent Corollary}[section]

\theoremstyle{remark}

\newtheorem{definition}{\bf \indent Definition}[section]

\numberwithin{equation}{section}

\def\con{\mathop{\mathstrut\rm Con}\nolimits}
\def\Con{\mathop{\mathstrut\rm CON}\nolimits}
\def\consc{\mathop{\mathstrut\textsc{Con}}\nolimits}
\newcommand{\low}{\mathopen\downarrow\,}
\newcommand{\dda}{\mathord{\mbox{\makebox[0pt][l]{\raisebox{-.4ex}
                           {$\downarrow$}}$\downarrow$\,}}}

\newcommand{\appmap}[3]{\mbox{$#1 \colon #2 \trianglelefteq #3$}}
\newcommand{\cfrel}[3]{\mbox{$#1 \colon #2 \bowtie #3$}}
\newcommand{\fun}[3]{\mbox{$#1 \colon #2 \rightarrow #3$}}
\newcommand{\set}[2]{\mbox{$\{\,#1 \mid #2 \,\}$}}
\newcommand{\fsubset}{\subseteq_\mathrm{fin}}
\newcommand{\pow}[1]{\mathcal{P}(#1)}
\newcommand{\powf}[1]{\mathcal{P}_{\mathrm{fin}}(#1)}

\newcommand{\bA}{\mathbb{A}}
\newcommand{\DD}{\mathbb{D}}
\newcommand{\EE}{\mathbb{E}}
\newcommand{\UU}{\mathbb{U}}

\newcommand{\CCC}{\cal{C}}
\newcommand{\DDD}{\cal{D}}
\newcommand{\EEE}{\cal{E}}
\newcommand{\FFF}{\cal{F}}

\newcommand{\FF}{\frak{F}}
\newcommand{\HF}{\frak{H}}
\newcommand{\LF}{\frak{L}}
\newcommand{\UF}{\frak{U}}
\newcommand{\VF}{\frak{V}}
\newcommand{\XF}{\frak{X}}
\newcommand{\YF}{\frak{Y}}
\newcommand{\ZF}{\frak{Z}}

\newcommand{\bt}{\mathbf{t}}
\newcommand{\bM}{\mathbf{M}}
\newcommand{\bT}{\mathbf{T}}
\newcommand{\bC}{\mathbf{C}}
\newcommand{\INF}{\mathbf{INF}}
\newcommand{\aINF}{\mathbf{aINF}}
\newcommand{\sINF}{\mathbf{sINF}}
\newcommand{\asINF}{\mathbf{asINF}}
\newcommand{\DOM}{\mathbf{DOM}}
\newcommand{\aDOM}{\mathbf{aDOM}}
\newcommand{\CFA}{\mathbf{CFA}}
\newcommand{\tCFA}{\mathbf{tCFA}}
\newcommand{\CFAM}{\mathbf{CFA}_{\bM}}
\newcommand{\tCFAM}{\mathbf{tCFA}_{\bM}}

\def\Id{\mathop{\mathstrut\rm Id}\nolimits}

\def\ID{\mathop{\mathstrut\mathcal{I}}}
\def\sp{\mathop{\mathstrut\rm sp}}
\def\st{\mathop{\mathstrut\rm st}}
\def\mmodels{\mathrel {||}\joinrel \Relbar}

\begin{document}

\title{Domains, Information Frames, Rough Sets:\\ An Equivalence of Categories
\thanks{This research has partially been supported by DFG grant no. 549144494.}}

\author{Dieter Spreen\\
 Department of Mathematics, University of Siegen\\
 {\small
spreen@math.uni-siegen.de}
}

\date{}
\maketitle

\begin{abstract}
A generalization of Scott's information systems~\cite{sco82} is presented that captures exactly all continuous domains. The global consistency predicate in Scott's definition is relativized. Now, for every atomic statement, there is a consistency predicate that states which finite sets of statements express information that is consistent with the given statement.
The category of information frames is shown to be equivalent to the category of domains. Moreover, the relationship with CF-approximation spaces introduced by Wu and Xu~\cite{wx23} is studied. The corresponding category is also shown to be equivalent with the category of information frames. This research achieves a refinement of the equivalence result of Wu and Xu of the category of CF-approximation spaces with the category of domains. 
\end{abstract}

\section{Introduction}\label{sect-intro}

Research in domain theory started with the work of Dana Scott~\cite{sc72,sc73,sc76,sc80,sc82} and, independently, Yuri L.\ Ershov~\cite{er71,er73,er75,er76}.  D.\ Scott
was interested in constructing a natural mathematical model for the type-free $\lambda$-calculus, whereas Y. L.\ Ershov wanted to develop a theory of partial computable functionals of finite type. Today, the field is a thriving branch of mathematics with applications in theoretical computer science, logic, and topology. In theoretical computer science it is particularly used for providing functional programming languages with a semantics that is independent of implementations. For this reason, classes of domains have been studied that are closed under the function space construction, that is, the space of domain operation respecting maps. As was shown by Jung~\cite[Corollary 10]{ju89,ju90} there are two maximal such classes.

Ershov's approach to domain theory is based on topology, whereas Scott introduced domains by means of partial orders. Both approaches are equivalent. Here, we follow the latter one: A domain is a partial order $(D, \sqsubseteq)$ so that least upper bounds of directed subsets exist. In addition, a domain is required to contain a basis, that is, a subset coming with an interpolative transitive binary relation $\ll$ compatible with the partial order such that every domain element $x$ can be obtained as least upper bound of the base elements $b \ll x$. The interpolation condition is a strengthening of density. Sets $B$ coming with an interpolative transitive binary relation are called abstract bases.

A basis  turns out to be a skeleton of a domain: every domain is isomorphic to the round ideal completion of an abstract basis~\cite[Proposition 2.2.25]{aj94}. In  other words: Domains can be represented by abstract bases. Here, by representation of domains any class of structures is meant so that every domain is isomorphic to a family of sets generated by a structure of the given class and ordered by set inclusion.

In \cite{wx23}, G. Wu and L. Xu provide a representation of domains which is based on concepts from rough set theory. It generalizes the round ideal construction. They introduce generalized approximation spaces with consistent family of finite subsets, as well as the associated morphisms: CF-approximable relations.  The central result says that every such space generates a domain and conversely, up to isomorphism, each domain can be obtained in this way. 

As mentioned already, there are two maximal classes of domains closed under the function space construction. One of them is the class of L-domains. In order to present an easy to understand representation of these domains,  the present author~\cite{sp21}  introduced L-information frames (called information frames in \cite{sp21}). An  L-information frame consists of a set $P$ of  atomic propositions, a family $(\con_{p})_{p \in P}$ of finite subsets of atomic propositions, a family $(\vdash_{p})_{p \in P}$ of relations with $\vdash_{p} \subseteq  \con_{p} \times P$, and a dense transitive binary relation $R$ on $P$. Each set $\con_{p}$ can be thought of as being the collection of those finite sets of atomic propositions that are consistent with proposition $p$, and the relation $\vdash_{p}$ is an entailment relation for the (rudimentary) logic determined by proposition $p$. These relations are required to be sound, and closed under the rules Weakening, Interpolation and Cut. Moreover, if $pRq$, then the logic determined by $q$ is a conservative extension of the logic determined by $p$.

Information frames introduced in the present paper are defined like L-information frames, but the conservativity requirement  is omitted. As we shall see, they represent exactly all domains: states of the logic, that is, finitely consistent and entailment-closed sets of atomic propositions, form a domain with respect to set inclusion, and, conversely, up to isomorphism every such domain can be obtained in this way. Approximable mappings will be introduced as appropriate morphisms, and it will be shown that the category of information frames with approximable mappings is equivalent to the category of domains with Scott continuous functions, that is, maps respecting directed least upper bounds.  

Furthermore, this paper examines the connection between information frames and the spaces introduced by Wu and Xu. We show that from any given information frame a generalized approximation space with a consistent family of finite subsets can be constructed and that conversely, any such space generates an information frame. The constructions will allow us to derive the equivalence of the category of information frames with approximable mappings and the category of generalized approximation spaces with a consistent family of finite subsets and  CF-approximable relations.

The constructions of Wu and Xu, of domains from generalized approximate spaces with consistent families of finite subsets and vice versa, are thus each split into two constructions with information frames as an intermediate structure. Each of the new constructions is very natural and interesting in its own right.

Natural requirements for information frames are presented so that the domains constructed from them are algebraic or pointed. The conditions are reproduced by the constructions in the opposite direction. In the case of generalized approximation spaces with a consistent set of finite subsets, an analogous procedure is followed.

Wu and Xu's theorem is then a consequence of the results in this paper. Due to the way we proceed with the special cases just mentioned, analogous equivalence results arise for these cases.

The paper is organized as follows: In Section~\ref{sect-dom} some notations are established and definitions and basic results of domain theory are reviewed. Section~\ref{sect-info} deals with information frames: the definition is given and the constructions of a domain from a given information frame and vice versa are presented. 

In Section~\ref{sect-appmap}, approximable mappings are introduced. They are the morphisms between information frames. As usual, the morphisms between domains are Scott-continuous functions. The constructions given in the previous section are extended to these morphisms and it is shown that the functors thus obtained between the category of information frames and approximable mappings and the category of domains and Scott-continuous functions establish an equivalence between these categories.

Section~\ref{sect-rset} contains basic definitions from rough set theory and the definition of generalized approximation spaces with consistent families of finite subsets, in short CF-approximation spaces. It is shown that each of these spaces leads to an information frame and vice versa. 

Finally, in Section~\ref{sect-CFapp} CF-approximable relations are introduced, the morphisms between CF-approximation spaces. The constructions of the previous section are extended to the morphisms and it is shown that the functors thus obtained establish an equivalence between the category of CF-approximation spaces and CF-approximable relations and the category of information frames with approximable mappings.

\section{Domains: basic definitions and results}\label{sect-dom}

For any set $A$, we write $X \fsubset A$ to mean that $X$ is finite subset of $A$. The collection of all subsets of $A$ will be denoted by $\pow{A}$ and that of all finite subsets by $\powf{A}$.

Let $\DD = (D, \sqsubseteq)$ be a poset. $\DD$ is \emph{pointed} if it contains a least element $\bot$. For an element $x \in D$, $\low x$ denotes the principal ideal generated by $x$, i.e., $\low x = \set{y \in D}{y \sqsubseteq x}$. A subset $S$ of $D$ is \emph{directed}, if it is nonempty and every pair of elements in $S$ has an upper bound in $S$. $\DD$ is a \emph{directed-complete partial order} (\emph{dcpo}), if every directed subset $S$ of $D$ has a least upper bound $\bigsqcup S$ in $D$.

Assume that $x, y$ are elements of a poset $D$. Then $x$ is said to \emph{approximate} $y$, written $x \ll y$, if for any directed subset $S$ of $D$ the least upper bound of which exists in $D$, the relation $y \sqsubseteq \bigsqcup S$ always implies the existence of some $u \in S$ with $x \sqsubseteq u$. Moreover, $x$ is \emph{compact} if $x \ll x$.  A subset $B$ of $D$ is a \emph{basis} of $\DD$, if for each $x \in D$ the set $\dda\!_B x = \set{u \in B}{u \ll x}$ contains a directed subset with least upper bound $x$. Note that the set of compact elements of $D$ is included in every basis of $\DD$.  A directed-complete partial order $\DD$ is said to be \emph{continuous} (or a \emph{domain}) if it has a basis and it is called \emph{algebraic} (or an \emph{algebraic domain}) if its compact elements form a basis. Standard references for domain theory and its applications are \cite{aj94,sh94,ac98,gh03}.

\begin{lemma}\label{lem-preordprop}
In a poset $\DD$ the following statements hold for all $u, x, y, z \in D$: \begin{enumerate}
\item\label{lem-preordprop-0} The approximation relation $\ll$ is transitive.
\item\label{lem-preordprop-1} $x \ll y \Rightarrow x \sqsubseteq y$.
\item\label{lem-preordprop-2} $u \sqsubseteq x \ll y \sqsubseteq z \Rightarrow u \ll z$.
\item\label{lem-preordprop-k} If $x \sqsubseteq y$, and  one of $x, y$ is compact, then  $x \ll y$.
\item\label{lem-preordprop-4} If $\DD$ has a least element $\bot$, then $\bot $ is compact.
\item\label{lem_preordprop-5} If $\DD$ is a domain with basis $B$, and $M \fsubset D$, then
\begin{equation*}
M \ll x \Rightarrow (\exists v \in B) M \ll v \ll x,
\end{equation*}
where $M \ll x$ means that $m \ll x$, for any $m \in M$.

\end{enumerate}
\end{lemma}

Property~\ref{lem_preordprop-5} is known as the \emph{interpolation law}. For a proof see \cite[Lemma 2.2.5]{aj94}.

\begin{definition}
Let $\DD$ and $\DD'$ be posets. A function $\fun{f}{D}{D'}$ is \emph{Scott continuous} if for any directed subset $S$ of $D$ with existing least upper bound, 
\[
\bigsqcup f[S] = f(\bigsqcup S).
\]
\end{definition}

With respect to the pointwise order the set $[D \to D']$ of all Scott continuous functions between two dcpo's $\DD$ and $\DD'$ is a dcpo again. Observe that it need not be continuous even if $\DD$ and $\DD'$ are. This is the case, however, if $\DD'$ is an L-domain \cite{aj94}.

\begin{definition}
A pointed domain $\DD$ is an \emph{L-domain}, if each pair $x, y \in D$ bounded above by $z \in D$ has a least upper bound $x \sqcup^z y$ in $\low z$.
\end{definition}
Note that in \cite{gh03} pointedness is not required.

As has been shown by Jung \cite{ju89}, the category $\mathbf{L}$ of L-domains is one of the two maximal Cartesian closed full subcategories of the category $\DOM_{\perp}$ of pointed domains and Scott continuous functions. The one-point domain is the terminal object in these categories and the categorical product $\DD \times \EE$  of two domains $\DD$ and $\EE$ is the Cartesian product of the underlying sets  ordered coordinatewise.

\section{Information frames}\label{sect-info}

In this section we introduce information frames and study their relationship with domains. 

An information frame consists of a Kripke frame $(A, R)$, the nodes of which are also called tokens. Associated with each node $i \in A$ is a consistency predicate $\con_{i}$ classifying the finite sets of tokens which are consistent with respect to node $i$, and an entailment relation $\vdash_{i}$ between $i$-consistent sets and tokens.

The conditions that have to be satisfied are grouped. There are requirements which consistency predicate and entailment relation of each single node have to meet, and conditions that specify their interplay for nodes related to each other by the accessibility relation. 

\begin{definition}\label{dn-infofr}
Let $A$ be a set,  $(\con_i)_{i \in A}$ be a family of subsets of $\powf{A}$, and $(\vdash_i)_{i \in A}$ be a family of relations $\vdash_i  \subseteq \con_i \times A$. For $i, j \in A$
set
\begin{equation*}
iRj \Leftrightarrow \{ i \} \in \con_{j}.
\end{equation*}
Then $\bA = (A, (\con_i)_{i \in A}, (\vdash_i)_{i \in A})$ is an \emph{information frame} if the following conditions hold, for all $a \in A$ and all finite subsets $X, Y$ of $A$:
\begin{enumerate}
\item Local conditions, for every $i \in A$:

\begin{enumerate}
\item\label{dn-infofr-1}
$\{i\} \in \con_i$ \hfill (self consistency),

\item\label{dn-infofr-2}
$Y \subseteq X \wedge X \in \con_i \Rightarrow Y \in \con_i$ \hfill (consistency preservation),

\hspace{-\leftmargin}
and, defining $X \vdash_i Y$ to mean that $X \vdash_i b$, for all $b \in Y$,

\item\label{dn-infofr-4}
$X \in \con_{i}\mbox{} \wedge X \vdash_i Y \Rightarrow Y \in \con_i$ \hfill (soundness),

\item\label{dn-infofr-5}
$X, Y \in \con_i\mbox{} \wedge Y \supseteq X \wedge X \vdash_i a \Rightarrow Y \vdash_i a$ \hfill (weakening),

\item\label{dn-infofr-6}
$X \in \con_i\mbox{} \wedge X \vdash_i Y \wedge Y \vdash_i a \Rightarrow X \vdash_i a$ \hfill (cut),

\end{enumerate}

\item Global conditions, for all $i, j \in A$

\begin{enumerate}

\item\label{dn-infofr-7}
$i R j \Rightarrow \con_i \subseteq \con_j$ \hfill (consistency transfer),

\item\label{dn-infofr-9}
$i R j \wedge X \in \con_i \mbox{} \wedge X \vdash_i a \Rightarrow X \vdash_j a$ \hfill (entailment transfer),

\item\label{dn-infofr-11}
$X \vdash_i Y \Rightarrow (\exists e \in A) (\exists Z \in \con_{e}) X \vdash_i (\{ e \} \cup Z) \wedge Z \vdash_{e} Y$ \hfill (interpolation).

\end{enumerate}

\end{enumerate}
\end{definition}

All requirements are very natural. Note that from Condition~\ref{dn-infofr}\eqref{dn-infofr-4}, that is, soundness, it particularly follows that for $i, j \in A$ and $X \in \con_{j}$,
\begin{equation}\label{eq-entcon}
X \vdash_{j} i \Rightarrow i R j.
\end{equation}
By Condition~\eqref{dn-infofr-2}, $\emptyset \in \con_{i}$, for all $i \in A$. Hence, Requirement~\eqref{dn-infofr-4} holds if $Y = \emptyset$. Moreover, in this case, also \eqref{dn-infofr-5} is true and \eqref{dn-infofr-6} is a consequence of \eqref{dn-infofr-5}. 

Sometimes a stronger version of Cut is needed which reverses the Interpolation Axiom.

\begin{lemma}\label{lem-strong6}
Let $\bA$ be an information frame. Then the following rule holds, for all $a, i, j \in A$, $X \in \con_{i}$ and $Y \in \con_{j}$,
\[
X \vdash_{i} (\{j\} \cup Y)  \wedge Y \vdash_{j} a \Rightarrow X \vdash_{i} a.
\]
\end{lemma}
\begin{proof}
Since $X\vdash_{i} j$, it follows from Axiom~\ref{dn-infofr}\eqref{dn-infofr-4} that $\{j\} \in \con_{i}$, whence $j R i$ by \eqref{eq-entcon}. As a consequence of Axioms~\ref{dn-infofr}\eqref{dn-infofr-9} we therefore have that $Y \vdash_{i} a$. Now, we can apply Axiom~\ref{dn-infofr}\eqref{dn-infofr-6} to obtain that $X \vdash_{i} a$.
\end{proof}

Note next that from the global interpolation property \ref{dn-infofr}\eqref{dn-infofr-11} we in particular obtain that every entailment relation $\vdash_{i}$ is interpolative.
\begin{lemma}\label{lem-locint} 
Let $\bA$ be an information frame. Then the following two statements hold for all $i \in A$ and $X, Y \fsubset A$ with $X \in \con_{i}$ and $X \vdash_{i} Y$:
\begin{enumerate}
\item\label{lem-locint-1}
$(\exists Z \in \con_{i}) X \vdash_{i} Z  \wedge Z \vdash_{i} Y$.

\item\label{lem-locint-2}
$(\exists e \in A) X \vdash_{i} e \wedge Y \in \con_{e}$.

\end{enumerate}
\end{lemma}
\begin{proof}
Suppose that $i \in A$ and $X, Y \fsubset A$ with $X \in \con_{i}$ and $X \vdash_{i} Y$. Then it follows with the global interpolation condition that there are $e \in A$ and $Z \in \con_{e}$ with $X \vdash_{i} \{e \} \cup Z$ and $Z \vdash_{e} Y$. By \eqref{eq-entcon} it follows from $X \vdash_{i} e$ that $e  R i$, with which we obtain by Conditions~\ref{dn-infofr}\eqref{dn-infofr-7} and  \eqref{dn-infofr-9} that $Z \vdash_{i} Y$. Since $Z \vdash_{e} Y$ we moreover have that $Y \in \con_{e}$, by soundness.
\end{proof}

As follows from \ref{dn-infofr}\eqref{dn-infofr-9}, entailment is preserved when moving from node $i$ to $j$. However, it might be that at level $j$, $X$ entails more tokens than at $i$. This, however, is not the case if the information frame is conservative.

\begin{definition}\label{dn-cst}
Let $\bA$ be an information frame.
\begin{enumerate}
\item\label{dn-cst-c}
$\bA$ is said to be \emph{conservative} if the following Condition~\eqref{cn-C} holds for all $a, i, j \in A$ and $X \in \con_{i}$:
\begin{equation}\label{cn-C}
iRj \wedge X \vdash_{j} a \Rightarrow X \vdash_{i} a. \tag{C}
\end{equation}

\item\label{dn--cst-al}
$\bA$  is said to be \emph{algebraic} if the following Condition~\eqref{cn-A} holds for all $i \in A$ and $X \in \con_{i}$:
\begin{equation}\label{cn-A}\tag{AL}
X \vdash_{i} \{i \} \cup X.
\end{equation}

\item\label{dn-cst-s}
$\bA$ is said to be \emph{strong} if the following Condition~\eqref{cn-S} holds for all $i \in A$ and $X \fsubset A$:
\begin{equation}\label{cn-S}
X \in \con_{i}\mbox{} \wedge X \neq \{i\} \Rightarrow \{ i \} \vdash_{i} X. \tag{S}
\end{equation}

\item\label{dn-cst-t}
A token $\bt \in A$ is called \emph{truth element} if the following Condition~\eqref{cn-T} holds for all $i \in A$:
\begin{equation}\label{cn-T}
\emptyset \vdash_{i} \bt. \tag{T}
\end{equation}
\end{enumerate}
\end{definition}

In the sequel we will show that the global states of an information frame form a domain.
\begin{definition}\label{dn-st}
Let $\bA$ be an information frame. A subset $x$ of $A$ is a \emph{(global) state} of $\bA$ if the following three conditions hold:
\begin{enumerate}
\item\label{dn-st-1}
$(\forall F \fsubset x) (\exists i \in x) F \in \con_{i}$ \hfill (finite consistency),

\item\label{dn-st-2}
$(\forall i \in x)(\forall X \fsubset x) (\forall a \in A) [X \in \con_{i}\mbox{} \wedge X \vdash_{i} a \Rightarrow a \in x]$ \hfill (closure under entailment),

\item\label{dn-st-3}
$(\forall a \in x) (\exists i \in x) (\exists X \fsubset x) X \in \con_{i}\mbox{} \wedge X \vdash_{i} a$ \hfill (completeness). 

\end{enumerate}
\end{definition}

By Condition~\ref{dn-st}\eqref{dn-st-1} global states $x$ are never empty: Choose $F = \emptyset$. Then $x$ contains some $i \in A$ with $\emptyset \in \con_{i}$.

Note that Conditions~\eqref{dn-st-1} and \eqref{dn-st-3} in Definition~\ref{dn-st} can be replaced by a single requirement.

\begin{proposition}\label{pn-stsing}
Let $\bA$ be an information frame  and $x$ be a subset of $A$. Then Conditions~\ref{dn-st}\eqref{dn-st-1} and \eqref{dn-st-3} are equivalent to the following statement:
\begin{equation}\tag{ST}\label{st}
(\forall F \fsubset x) (\exists i \in x) (\exists X \fsubset x) X \in \con_{i}\mbox{} \wedge X \vdash_{i} F.
\end{equation}
\end{proposition}
The proofs of this and the following results are as in \cite[Lemmas 3, 6, 7, 12]{sp21}.

With respect to set inclusion the global states of $\bA$ form a partially ordered set, denoted by $|\bA|$.

\begin{lemma}\label{lem-infosysdom}
$|\bA|$ is directed-complete.
\end{lemma}

As we will see next, the consistent subsets of $\bA$ generate a canonical basis of $|\bA|$. For $i \in A$ and $X \in \con_{i}$ let
\begin{equation*}
[X]_i = \set{a \in A}{X \vdash_{i} a}.
\end{equation*}
\begin{lemma}\label{lem-canba}
\begin{enumerate}
\item\label{lem-canba-1} $[X]_i$ is a global state of $\bA$, for each $i \in A$ and $X \in \con_{i}$.
\item\label{lem-canba-2} For every $z \in |\bA|$, the set of all $[X]_i$ with  $\{i\} \cup X \subseteq z$ is directed and $z$ is its union.
\end{enumerate}
\end{lemma}

This result allows characterizing the approximation relation on $|\bA|$ in terms of the entailment relation. The characterization nicely reflects the intuition that $x \ll y$ if $x$ is covered by a ``finite part'' of $y$.

\begin{proposition}\label{pn-app}
For $x, y \in |\bA|$,
\begin{equation*}
x \ll y \Leftrightarrow (\exists i \in A)(\exists V \in \con_{i}) \{i\} \cup V \subseteq y \wedge V \vdash_{i} x.
\end{equation*}
\end{proposition}

Note that by Axioms~\ref{dn-infofr}\eqref{dn-infofr-1} and \eqref{dn-infofr-2} $\emptyset \in \con_{i}$, for all $i \in A$. Suppose now that $\bA$ has a truth element $\bt$. Then it follows with Condition~\eqref{cn-T} and Axiom~\ref{dn-infofr}\eqref{dn-infofr-4} that $\{ \bt \} \in \con_{j}$, for all $j \in A$.

\begin{lemma}\label{lem-point}
Let $\bA$ have a truth element $\bt$. Then $[\emptyset]_{\bt}\subseteq x$, for all $x \in |\bA|$.
\end{lemma}
\begin{proof}
As states are nonempty, there is some $i \in x$. Moreover, by Condition~\eqref{cn-T}, $\emptyset \vdash_{i} \bt$. Thus, $\bt \in x$, because of \ref{dn-st}\eqref{dn-st-2}. Hence, by applying the same rule again, we obtain that $[\emptyset]_{\bt} \subseteq x$.
\end{proof} 

Let us now sum up what we have reached so far.
\begin{theorem}\label{thm-infdom}
Let $\bA$ be an information frame. Then $\DDD(\bA) = (|\bA|, \subseteq)$ is a domain with basis $\hat{\Con} = \set{[X]_{i}}{i \in A \wedge X \in \con_{i}}$. Morevover,
\begin{enumerate}
\item\label{thm-infdom-L}\
$\DDD(\bA)$ is an L-domain, if $\bA$ is conservative.

\item\label{thm-infdom-AL}
$\DDD(\bA)$ is algebraic, if $\bA$ is algebraic.

\item\label{thm-infdom-T}
$\DDD(\bA)$ is pointed with least element $[\emptyset]_{\bt}$, if $\bA$ has a truth element $\bt$.

\end{enumerate}
\end{theorem}
\begin{proof}
\eqref{thm-infdom-L} is shown in \cite[Theorem 11]{sp21}. For \eqref{thm-infdom-AL} note that by Proposition~\ref{pn-app} all states in $\hat{\Con}$ are compact.
\end{proof}

Now, conversely, let $\DD$ be a domain with basis $B$. Set 
\begin{equation*}
\FFF(\DD) = (B, (\Con_{i})_{i \in B}, (\vDash_{i})_{i \in B})
\end{equation*}
 with
\begin{align*}
&\Con_{i} =\{ \{ i \}\} \cup \powf{\dda_{B} i},  \\
&X \vDash_{i} a \Leftrightarrow (\exists b \in X \cup \{i\}) a \ll b.
\end{align*}
\begin{theorem}\label{thm-dominf}
Let $\DD$ be a domain. Then $\FFF(\DD)$ is a strong information frame such that 
\begin{enumerate}
\item\label{thm-dominf-1}
$\FFF(\DD)$ is algebraic, if $\DD$ is algebraic and all $z \in B$ are compact.

\item\label{thm-dominf-2}
$\FFF(\DD)$  has a truth element, if $\DD$ is pointed.
\end{enumerate}
\end{theorem}
\begin{proof}
We only verify Condition~\eqref{dn-infofr-11} in Definition~\ref{dn-infofr}; the other ones being obvious. Assume that $X \vDash_{i} Y$, for $i \in B$, $Y \fsubset B$ and $X \in  \Con_{i}$. Then $Y \ll i$. By the interpolation law there is thus some $e \in B$ with $Y \ll e \ll i$. Set $Z = \{e\}$. Then $Z \in \Con_{e}$. Moreover, $X \vDash_{i} \{e\} \cup Z$ and $Z \vDash_{e} Y$.

 Condition~\eqref{cn-S} is obvious as well. Therefore, $\FFF(\DD)$ is strong. Statement~\eqref{thm-dominf-1} is also obvious as all $z \in B$ are compact. For~\eqref{thm-dominf-2} let $\bot$ be the least element of $\DD$. Then $\bot \in B$ and $\bot \ll i$, for every $i \in B$, again by Lemma~\ref{lem-preordprop}. Thus, $\bot$ is a truth element of $\FFF(\DD)$.
\end{proof}

\section{Approximable mappings}\label{sect-appmap}

In the next step we want to turn the collection of information frames into a category.  The appropriate morphisms are  families of relations similar to entailment relations. 

 \begin{definition}\label{dn-am}
 Let  $\bA$ and $\bA'$ be information frames. 
 \begin{enumerate}
 \item
 An \emph{approximable mapping}  between $\bA$ and $\bA'$,  written $\appmap{H = (H_{i})_{i \in A}}{\bA}{\bA'}$, is a family of relations $H_{i} \subseteq \con_{i} \times A'$ satisfying for all $X, X' \in \con_{i}$, $Y \in \bigcup_{p \in A'} \con'_{p}$ and $b, k \in A'$, as well as all  $F \fsubset A'$ the following five conditions, where $XH_{i}Y$ means that $XH_{i}c$, for all $c \in Y$:
 \begin{enumerate}
 \item \label{dn-am-1} $X H_{i} (\{k\} \cup Y) \wedge Y \vdash'_{k} b \Rightarrow X H_{i} b$,
 
 \item\label{dn-am-2} $X \subseteq X' \wedge X H_{i} b \Rightarrow X' H_{i} b$,
 
 \item\label{dn-am-3} $X \vdash_{i} X' \wedge X' H_{i} b \Rightarrow X H_{i} b$,
 
 \item \label{dn-am-4}$i R j \wedge X H_{i} b \Rightarrow X H_{j} b$,
 
 \item\label{dn-am-5} $X H_{i} F \Rightarrow \\
 \mbox{}\hfill(\exists c \in A) (\exists e \in  A') (\exists U \in \con_{c}) (\exists V \in \con'_{e})  [X \vdash_{i} (\{c\} \cup U) \wedge U H_{c} (\{e \} \cup V) \wedge V \vdash'_{e} F]$.
 
 \end{enumerate}
 
 \item Let $\bA$ and $\bA'$, respectively, have truth elements $\bt$ and $\bt'$. Then $\appmap{H}{\bA}{\bA`}$ \emph{respects truth elements}, if 
 
 \begin{enumerate}
 
 \item\label{dn-am-6}  $\emptyset H_{\bt} \bt'$.
 
 \end{enumerate}
 \end{enumerate}
 \end{definition}

In applications it is sometimes preferable to have Condition~\eqref{dn-am-5} split up into two conditions which state interpolation for the domain and the range of the approximable mapping, separately.

\begin{lemma}\label{pn-amint}
Let $\bA$ and $\bA'$ be information frames. Then, for any family $(H_{i})_{i \in A}$ with $H_{i} \subseteq \con_{i} \times A'$, $X \in \con_{i}$, and $F \fsubset A'$, Condition~\ref{dn-am}\eqref{dn-am-5} is equivalent to the following Conditions~\eqref{pn-amint-1} and \eqref{pn-amint-2}:
\begin{enumerate}
\item\label{pn-amint-1}
$X H_{i} F \Rightarrow (\exists c \in A)(\exists U \in \con_{c}) X \vdash_{i} (\{c\} \cup U) \wedge U H_{c} F$

\item\label{pn-amint-2}
$X H_{i} F \Rightarrow (\exists e \in A')(\exists V \in \con'_{e}) X H_{i} (\{e\} \cup V) \wedge V \vdash'_{e} F$.
\end{enumerate}
\end{lemma}
Again, The proofs of this and the next result are as in \cite[Lemmas 25, 26]{sp21}.

Moreover, a strengthening of Condition~\ref{dn-am}\eqref{dn-am-3} can be derived which reverses the implication in the first statement of the preceding lemma.
\begin{lemma}\label{lem-amstrong3}
Let $H$ be an approximable mapping between information frames  $\bA$ and $\bA'$ . Then for all $i, j \in A$, $X \in \con_{i}$, $Y \in \con_{j}$ and $b \in A'$,
\[
 X \vdash_{i} (\{j\} \cup Y) \wedge Y H_{j} b \Rightarrow X H_{i} b.
\]
\end{lemma}

For $\nu = 1,2,3$, let $\bA^{(\nu)}$ be information frames, and $\appmap{G}{\bA^{(1)}}{\bA^{(2)}}$ and $\appmap{H}{\bA^{(2)}}{\bA^{(3)}}$. Define $G \circ H = ((G \circ H)_{i})_{i \in A^{(1)}}$ with
\begin{equation*}
X (G \circ H)_{i} a \Leftrightarrow (\exists e \in A^{(2)})(\exists V \in \con^{(2)}_{e}) X G_{i} (\{e\} \cup V) \wedge  V H_{e} a,
\end{equation*}
for $i \in A^{(1)}$, $X \in \con^{(1)}_{i}$ and $a \in A^{(3)}$. 

\begin{lemma}\label{lem-amprop}
For $\nu = 1, 2, 3$, let $\bA^{(\nu)}$ be information frames and $\appmap{G}{\bA^{(1)}}{\bA^{(2)}}$ as well as $\appmap{H}{\bA^{(2)}}{\bA^{(3)}}$. Then the following four statements hold:
\begin{enumerate}

\item\label{lem-amprop-1}
 $\appmap{(\vdash_{i})_{i \in A}}{\bA}{\bA}$.

\item\label{lem-amprop-2}
$\appmap{G \circ H}{\bA^{(1)}}{\bA^{(3)}}$. 

\item\label{lem-amprop-3}
If $G$ and $H$ respect existing truth elements, the same does $G \circ H$.

\item\label{lem-amprop-4}
$(\vdash^{(1)}_{i})_{i \in A^{(1)}} \circ G = G \circ (\vdash^{(2)}_{j})_{j \in A^{(2)}} = G$.

\end{enumerate}
\end{lemma}

In the sequel let $\INF$ be the category of information frames and approximable mappings, and \textbf{aINF},  \textbf{sINF} as well as \textbf{asINF}, respectively, be the full subcategories of algebraic, of strong, and of algebraic strong information frames. Moreover, let $\INF_{\bt}$ be the subcategory of information frames with truth elements as objects and truth element respecting approximable mappings as morphisms, and $\aINF_{\bt}$, $\sINF_{\bt}$ as well as $\asINF_{\bt}$, respectively, be the full subcategories of algebraic, of strong and of algebraic strong information frames with truth element.

Similarly, let $\DOM$ denote the category of domains and Scott continuous functions and \textbf{aDOM}, $\DOM_{\bot}$, and  $\aDOM_{\bot}$, respectively, be the full subcategories of algebraic domains, pointed domains and algebraic pointed domains. 

By the preceding lemma  $\Id_{\bA} = (\vdash_{i})_{i \in A}$ is the identity morphism on information frame $\bA$. We have already shown that information frames lead to domains and vice versa. We will now do the same for approximable mappings and Scott continuous functions.

Let  $\bA$ and $\bA'$ be information frames and $\appmap{H}{\bA}{\bA'}$.

\begin{lemma}\label{lem-apsc}
For  $x \in |\bA|$, 
\[
\set{a \in A'}{(\exists i \in A)(\exists X \in \con_{i}) \{ i \} \cup X \subseteq x \wedge X H_{i} a} \in |\bA'|.
\]
\end{lemma}
The proof is as in \cite[Lemma 30]{sp21}.
The result allows us to define a function $\DDD(H) : \DDD(\bA) \to \DDD(\bA')$ by
\[
\DDD(H)(x) =  \{\, a \in A' \mid (\exists i  \in A) (\exists X \in \con_{i}) \{ i \} \cup X \subseteq x \wedge X H_{i} a \,\}.
\]

\begin{lemma}\label{lem-apmapst}
$\DDD(H)$ is Scott continuous.
\end{lemma}
The proof of this and the next lemma is as in \cite[Lemmas 31, 32]{sp21}.

\begin{lemma}\label{lem-funcl}
$\fun{\DDD}{\INF}{\DOM}$ is a functor such that $\DDD[\aINF] \subseteq \aDOM$, $\DDD[\INF_{\bt}] \subseteq \DOM_{\bot}$, and $\DDD[\aINF_{\bt}] \subseteq \aDOM_{\bot}$
\end{lemma}

Next, let us consider the reverse situation in which we move from domains to information frames. As we will see, every Scott continuous function $\fun{f}{D}{D'}$ between domains $\DD$ and $\DD'$ defines an approximable mapping $\appmap{\FFF(f)}{\FFF(\DD)}{\FFF(\DD')}$. 

Let $\DD$ and $\DD'$, respectively, have bases $B$ and $B'$. Then, for $i \in B$, $X \in \{\{i\}\} \cup \powf{\dda_{B} i}$, and $a \in B'$, set
\[
X \FFF(f)_{i} a \Leftrightarrow (\exists c \in X \cup \{i\}) a \ll' f(c).
\]

\begin{lemma}\label{lem-scam}
$\appmap{\FFF(f)}{\FFF(\DD)}{\FFF(\DD')}$. Moreover, $\FFF(f)$ respects truth elements in case $\DD$ and $\DD'$ are pointed.
\end{lemma}
\begin{proof}
Let $i \in B$, $X \in \{\{i\}\} \cup \powf{\dda_{B} i}$, $k, d \in B'$, and $Y, S \fsubset A'$. We have to verify Conditions~\ref{dn-am}\eqref{dn-am-1}-\eqref{dn-am-6}.

\eqref{dn-am-1}  Assume that $X \FFF(f) _{i }(\{k\} \cup Y)$ and $Y \vDash'_{k} b$. Then we have that for all $d \in Y \cup \{k\}$ there is some $c_{d} \in X \cup \{i\}$ with $d \ll' f(c_{d})$. Moreover, there is some $\hat{d}  \in Y \cup \{k\}$ with $b \ll' \hat{d}$. Therefore, $b \ll' f(c_{\hat{d}})$. That is, $X \FFF(f)_{i} b$.

Condition~\eqref{dn-am-3} follows in a similar way, and Conditions~\eqref{dn-am-2} and \eqref{dn-am-4} are obvious, as is Condition~\eqref{dn-am-6} in case that $\DD$ and $\DD'$ are pointed. We consider only Condition~\eqref{dn-am-5}.

Assume that $X \FFF(f)_{i} S$. Then there is some $c_{s} \in X \cup \{i\}$ with  $s \ll' f(c_{s})$, for each $s \in S$. By the monotonicity of $f$ and Lemma~\ref{lem-preordprop} we have  in particular  that $S \ll' f(i)$.  Because of the interpolation law there is thus some $e \in B'$ with $S \ll' e \ll' f(i)$.  As $i = \bigsqcup \dda_{B} i$ and $f$ is Scott continuous, we have that $f(i) = \bigsqcup f[\dda_{B} i]$. It follows that there is some $d \in B$ with $d \ll i$ so that $e \ll' f(d)$. Set $U = \{d\}$ and $V = \{e\}$. Then we have that $X \vDash_{i} \{d\} 
\cup U$, $U \mathcal{F}(f)_{d} (\{e\}\cup V)$ and $V \vDash'_{e} S$.
\end{proof}

\begin{lemma}\label{lem-funci}
$\fun{\FFF}{\DOM}{\sINF}$ is a functor such that 
\begin{enumerate}
\item
$\FFF[\aDOM] \subseteq \asINF$, 

\item
$\FFF[\DOM_{\bot}] \subseteq \sINF_{\bt}$ and 

\item
$\FFF[\aDOM_{\bot}] \subseteq \asINF_{\bt}$.

\end{enumerate}
\end{lemma}

Functoriality follows from continuity, similarly as in the last step of the preceding proof. The other property is obvious.

As we will show next, the functors $\fun{\DDD}{\INF}{\DOM}$ and $\fun{\FFF}{\DOM}{\INF}$ establish an equivalence between the categories $\INF$ and $\DOM$.

For a category $\bC$ let $\ID_{\bC}$ be the identity functor on $\bC$. We first construct a natural isomorphism $\fun{\eta}{\ID_{\INF}}{\FFF \circ \DDD}$. Let to this end $\bA$ be an information frame. Then 
\begin{equation*}
\FFF( \DDD(\bA)) = (\hat{\con}, (\Con_{u})_{u \in \hat{\con}}, (\vDash_{u})_{u \in \hat{\con}}).
\end{equation*} 

For $i, a \in A, X \in \con_{i}, u \in \hat{\con}$ and  $\XF \in \Con_{u}$ define  
\begin{align*}
&X S^{A}_{i} u \Leftrightarrow u \ll [X]_{i}, \\
&\XF T^{A}_{u} a \Leftrightarrow (\exists v \in \hat{\con}) \XF \vDash_{u} v \wedge a \in v
\end{align*}
and set $S_{A} = (S^{A}_{i})_{i \in A}$ and $T_{A} = (T^{A}_{u})_{u \in \hat{\con}}$.

\begin{lemma}\label{lem-IF}
\begin{enumerate}

\item\label{lem-IF-1}
$\appmap{S_{\bA}}{\bA}{\FFF(\DDD(\bA))}$ such that truth elements are respected, if $\bA$ has a truth element.

\item\label{lem-IF-2}
$\appmap{T_{\bA}}{\FFF(\DDD(\bA))}{\bA}$ such that truth elements are respected, if $\bA$ has a truth element.

\item\label{lem-IF-3}
$S_{\bA} \circ T_{\bA} = \Id_{\bA}$.

\item\label{lem-IF-4}
$T_{\bA} \circ S_{\bA} = \Id_{\FFF(\DDD(\bA))}$.

\end{enumerate}
\end{lemma}
\begin{proof}
\eqref{lem-IF-1}
We have to verify the conditions in Definition~\ref{dn-am}. Let to this end $i \in A$ and $X \in \con_{i}$. 

For Condition~\eqref{dn-am-1} assume further that $u, v \in  \hat{\con}$ and $\XF \in \Con_{u}$ with $X S^{A}_{i} (\{u\} \cup \XF)$ and $\XF \vDash_{u} v$. By the first assumption we have that for all $w \in \{u\} \cup \XF$, $w \ll [X]_{i}$, and by the second one that for some $w' \in \{u\} \cup \XF$, $v \ll  w'$. It follows that $v \ll [X]_{i}$, that is, $X S^{A}_{i} v$,

Next, for Condition~\eqref{dn-am-2} let $X' \in \con_{i}$ with $X \subseteq X'$. It follows that $[X]_{i} \subseteq [X']_{i}$. Now, assume in addition that  for $u \in \hat{\con}$, $X S^{A}_{i} u$. Then $u \ll [X]_{i} \subseteq [X']_{i}$. Hence, $u \ll [X']_{i}$, by Lemma~\ref{lem-preordprop}\eqref{lem-preordprop-2}. Thus, $X' S^{A}_{i} u$.

Conditions~\eqref{dn-am-3} and \eqref{dn-am-4} follow similarly as $[X']_{i} \subseteq [X]_{i}$ if $X \vdash_{i} X'$, and $[X]_{i} \subseteq [X]_{j}$ if $i R j$. 

For Condition~\eqref{dn-am-5} let $\YF \subseteq \hat{\con}$ with $X S^{A}_{i} \YF$. Then $\YF \ll [X]_{i}$. With the interpolation law it follows that there are $w, w' \in \hat{\con}$ so that $\YF \ll w \ll w' \ll [X]_{i}$. Hence, by Proposition~\ref{pn-app}, there is some $[V]_{k} \in \hat{\con}$ with $\{k\} \cup V \subseteq [X]_{i}$ and $V \vdash_{k} w'$. The latter implies that $w' \subseteq [V]_{k}$. Thus, $w \ll [V]_{k}$, by Lemma~\ref{lem-preordprop}\eqref{lem-preordprop-2}. So, we have that $X \vdash_{i} (\{k\} \cup V)$ and $V S^{A}_{k} w$. Since, $\YF \ll w$, we also have that $\{w\} \vDash_{w} \YF$.

Finally, for Condition~\eqref{dn-am-6}, assume that $\bA$ has truth element $\bt$. Then $\FFF(\DDD(\bA))$ has truth element $[\emptyset]_{\bt}$. As we have seen in Lemma~\ref{lem-point}, $[\emptyset]_{\bt}$ is the least element of $\DDD(\bA)$ and therefore compact by Lemma~\ref{lem-preordprop}\eqref{lem-preordprop-4}. Thus, $\emptyset S^{A}_{\bt} [\emptyset]_{\bt}$.

\eqref{lem-IF-2}
We only verify Condition~\ref{dn-am}\eqref{dn-am-1}. The other conditions are obvious.

Assume that $u \in \hat{\con}$, $\XF \in \Con_{u}$,  $e \in A$ and $U \in \con_{e}$ so that $\XF T^{\bA}_{u} (\{e\} \cup U)$ and $U \vdash_{e} a$. Then there are  $v_{k} \in \hat{\con}$ with $\XF \vDash_{u} v_{k}$ and $k \in v_{k}$, for all $k \in \{e\} \cup U$.  Set $\ZF = \set{v_{k}}{k \in \{e\} \cup U}$. Then $\ZF \ll u$, by the definition of $\vDash_{u}$. Since $\ZF$ is finite, there is some $w \in \hat{\con}$ with $\ZF \ll w \ll u$, by the interpolation law. With Lemma~\ref{lem-preordprop} it follows that $\{e\} \cup U \subseteq \bigcup \ZF \subseteq w$.  As $w \ll u$, we have that $\XF \vDash_{u} w$. Because $w \in \hat{\con}$, there are $j \in A$ and $V \in \con_{j}$ so that $w = [V]_{j}$. Thus, we have that $V \vdash_{j} \{e\} \cup U$. By assumption, $U \vdash_{e} a$. With Lemma~\ref{lem-strong6} we therefore have that $V \vdash_{j} a$, that is, $a \in w$. This shows that $\XF T^{\bA}_{u} a$.

\eqref{lem-IF-3}
Let $i, a \in A$ and $X \in \con_{i}$ with $(S_{\bA} \circ T_{\bA})_{i} a$. Then there is some $u \in \hat{\con}$ and $\XF \in \Con_{u}$ so that $X S^{\bA}_{i} (\{u\} \cup \XF)$ and $\XF T^{\bA}_{u} a$. That is, $(\{u\} \cup  \XF) \ll [X]_{i}$ and for some $v \in \hat{\con}$ with $a \in v$, $\XF \vDash_{u} v$. By the definition of $\vDash_{u}$ there is then some $w \in \XF \cup \{u\}$ with $v \ll w$. Hence, $v \ll [X]_{i}$ and thus $v \subseteq [X]_{i}$. It follows that $X \vdash_{i} a$, that is, $X \Id_{\bA} a$. 

For the converse inclusion, assume that $X \Id_{\bA} a$, that is, $X \vdash_{i} a$. By the interpolation property there exists $e, k \in A$, $U \in \con_{e}$ and $V \in \con_{k}$ so that $X \vdash_{i} \{e\} \cup U$, $U \vdash_{e} \{k\} \cup V$  and $V \vdash_{k} a$. Then $a \in [V]_{k}$. Moreover, it follows with Proposition~\ref{pn-app} that $[V]_{k} \ll [U]_{e} \ll [X]_{i}$. Set $u = [U]_{e}$ and $\XF = \{u\}$. Then $\XF \vDash_{u} [V]_{k}$ and hence $\XF T^{\bA}_{u} a$. In addition, $X S^{\bA}_{i} (\{u\} \cup \XF)$. Therefore, $X (S_{\bA} \circ T_{\bA})_{i} a$.

\eqref{lem-IF-4}
Let $u, v \in \hat{\con}$ and $\XF \in \Con_{u}$ with $\XF (T_{\bA} \circ S_{\bA})_{u} v$. Then there are $i \in A$ and $X \in \con_{i}$ so that $\XF T^{\bA}_{u} (\{i\} \cup X)$ and $X S^{\bA}_{i} v$. It follows for some $w \in \hat{\con}$ that $\XF \vDash_{u} w$ and $\{i\} \cup X \subseteq w$. Hence, $[X]_{i} \ll w$ and $v \ll [X]_{i}$. So, $v \ll w$, that is, $\{w\} \vDash_{w} v$. With Lemma~\ref{lem-strong6} it follows that $\XF \vDash_{u} v$.

Conversely, let $\XF \vDash_{u} v$. Then there is some $w \in \XF \cup \{u\}$ with $v \ll w$. Furthermore, by the interpolation law, there are $x, z \in \hat{\con}$ so that $v \ll  x \ll z \ll w$, which means we have that $\XF \vDash_{u} z$ and $v \ll x \ll z$. By Proposition~\ref{pn-app} it follows from the latter that are $e \in A$ and $V \in \con_{e}$ with $\{e\} \cup V \subseteq z$ and $V \vdash_{e} x$. Hence, $v \ll x \subseteq [V]_{e}$, which implies that $v \ll [V]_{e}$. Thus, we have that $\XF T^{\bA}_{u} (\{e\} \cup V)$ and $V S^{\bA}_{e} v$, that is, $\XF (T_{\bA} \circ S_{\bA})_{u} v$. 
\end{proof}

Set $\eta_{\bA} = S_{\bA}$. We want to show that $\eta$ is a natural transformation. 

\begin{lemma}\label{lem-etanat}
Let $\bA'$ be a further information frame and $\appmap{H}{\bA}{\bA'}$.
\begin{enumerate}

\item \label{lem-etanat-1}
Let $u \in \hat{\con}$, $\XF \in \Con_{u}$ and $v \in \hat{\con'}$. Then 
\begin{align*}
\XF \FFF(\DDD(H))_{u} v \Leftrightarrow &(\exists w \in \{u\} \cup \XF) (\exists j \in A) (\exists Z \in \con_{j}) \\
&(\exists k \in A') (\exists V \in \con'_{k}) \{j\} \cup Z \subseteq w \wedge Z H_{j} (\{k\} \cup V) \wedge V \vdash'_{k} v.
\end{align*}
 
\item\label{lem-etanat-2}
$\eta$ is a natural transformation.

\end{enumerate}
\end{lemma}
\begin{proof}
\eqref{lem-etanat-1} Assume that $\XF \FFF(\DDD(H))_{u} v$. Then there is some $w \in \{u\} \cup \XF$ so that $v \ll' \DDD(H)(w)$. With Proposition~\ref{pn-app} it follows that there are  $k \in A'$ and $V \in \con'_{k}$ with $v \subseteq [V]_{k}$ and $\{k\} \cup V \subseteq \DDD(H)(w)$. By the definition of $\DDD(H)(w)$ we now obtain that for every $a \in \{k\} \cup V$ there is some $i_{a} \in A$ and some $U_{a} \in \con_{i_{a}}$ such that $\{i_{a}\} \cup U_{a} \subseteq w$ and $U_{a} H_{i_{a}} a$. Set $Z= \bigcup\set{\{i_{a}\} \cup U_{a}}{a \in \{k\} \cup V}$. Then $Z \fsubset w$. Hence there is some $j \in w$ with $Z \in \con_{j}$. Since $i_{a} \in Z$, for all $a \in \{k\} \cup V$, we have that $U_{a} H_{j} a$ and furthermore that $Z H_{j} (\{k\} \cup V)$. As already seen, $V \vdash_{k} v$.  

For the converse implication let  $w \in \{u\} \cup \XF$, $ j \in A$,  $Z \in \con_{j}$, $ k \in A'$ and $V \in \con'_{k}$ with $\{j\} \cup Z \subseteq w$, $Z H_{j} (\{k\} \cup V)$ and $V \vdash'_{k} v$. Then $ \{k\} \cup V \subseteq \DDD(H)(w)$. Moreover, by Proposition~\ref{pn-app}, $v \ll' \DDD(H)(w)$. Thus,  $\XF \FFF(\DDD(H))_{u} v$.

\eqref{lem-etanat-2} We only have to show that
\begin{equation*}
\eta_{\bA} \circ \FFF(\DDD(H)) = H \circ \eta_{\bA'}.
\end{equation*}
Let $i \in A$, $X \in \con_{i}$ and $v \in \hat{\con'}$ with $X (S_{\bA} \circ \FFF(\DDD(H)))_{i} v$. Then there are  $u \in \hat{\con}$ and $\XF \in \Con_{u}$ with $X S^{\bA}_{i} (\{u\} \cup \XF)$ and $\XF (\FFF(\DDD(H)))_{u} v$. Then $\{u\} \cup \XF \ll [X]_{i}$. Moreover, with Statement~\eqref{lem-etanat-1} it follows that there are $w \in \{u\} \cup \XF$, $j \in A$,  $Z \in \con_{j}$, $k \in A'$ and $V \in \con'_{k}$ with $\{j\} \cup Z \subseteq w$, $Z H_{j} (\{k\} \cup V)$ and $V \vdash'_{k} v$. Hence, $w \ll [X]_{i}$.  As $\{j\} \cup Z \subseteq w \subseteq [X]_{i}$, we obtain that $X \vdash_{i} (\{j\} \cup Z)$. Therefore, $\{j\}\in \con_{i}$, which implies that $Z H_{i} (\{k\} \cup V)$.  Consequently, $X H_{i} (\{k\} \cup V)$. Because of Lemma~\ref{pn-amint}\eqref{pn-amint-2} there are $e \in A'$ and $Y \in \con'_{e}$ with $X H_{i} (\{e\} \cup Y)$ and $Y \vdash'_{e} \{k\} \cup V$. Since $V \vdash'_{k} v$,  we have that $v \ll' [Y]_{e}$, that is, $Y S^{A'}_{e} v$. Thus, $X (H \circ S_{\bA'})_{i} v$.

For the converse implication let again  $i \in A$, $X \in \con_{i}$ and $v \in \hat{\con'}$ with $X (H \circ S_{\bA'})_{i} v$. Then there is exist $k \in A'$ and $V \in \con'_{k}$ so that $X H_{i} (\{k\} \cup V)$ and $v \ll' [V]_{k}$. Then $v \subseteq [V]_{k}$ and hence $V \vdash'_{k} v$.  As a consequence of the first property we obtain with Lemma~\ref{pn-amint}\eqref{pn-amint-1} that there are $c \in A$ and $U \in \con_{c}$ so that $X \vdash_{i} (\{c\} \cup U)$ and $U H_{c} (\{k\} \cup V)$. Moreover, because of the interpolation property, there are $j \in A$ and $Z \in \con_{j}$ with $X \vdash_{i} (\{j\} \cup Z)$ and $Z \vdash_{j} (\{c\} \cup U)$. Set $w = [Z]_{j}$. Then $w \in \hat{\con}$, $\{c\} \cup U \subseteq w$ and $w \ll [X]_{i}$, the latter by Proposition~\ref{pn-app}. Define $\XF = \{w\}$. With Statement~\eqref{lem-etanat-1} it follows that $\XF \FFF(\DDD(H))_{w} v$. Since moreover $\{w\} \cup \XF \ll [X]_{i}$, we also have that $X S^{A}_{i} (\{w\} \cup \XF)$. This shows that $X (S_{\bA} \circ H)_{i} v$.
\end{proof}

Let us summarize what we have just shown.

\begin{proposition}\label{pn-etaiso}
$\fun{\eta}{\ID_{\INF}}{\FFF \circ \DDD}$ is a natural isomorphism.
\end{proposition}

Next, we show that there is also a natural isomorphism $\fun{\tau}{\ID_{\DOM}}{\DDD \circ \FFF}$. Let to this end $\DD$ be a domain. Then $\DDD(\FFF(\DD))$ is the domain $(|\FFF(\DD)|, \subseteq)$ with basis $\hat{\con}$, where $\FFF(\DD)$ is the information frame $(B, (\Con_{i})_{i \in B}, (\vDash_{i})_{i \in B})$.

\begin{lemma}[Lemma 18 in \cite{sp21}]\label{lem-stdir}
Every state of $\FFF(\DD)$ is a directed subset of $\DD$.
\end{lemma}

It follows that $\bigsqcup x$ exists in $D$, for every $x \in |\FFF(\DD)|$. For $x \in |\FFF(\DD)|$ set 
\[
\sp\nolimits_\DD(x) = \bigsqcup x.
\]
Then $\fun{\sp_\DD}{|\FFF(\DD)|}{D}$ is Scott continuous.

\begin{lemma}[Lemma 19 in \cite{sp21}]\label{lem-domst}
For $x \in D$, $\set{a \in B}{a \ll x}$ is a state of $\FFF(\DD)$.  
\end{lemma}

Set 
\[
\st\nolimits_\DD(x) = \set{a \in B}{a \ll x},
\]
for $x \in D$. Then $\fun{\st_\DD}{D}{|\FFF(\DD)|}$ is Scott continuous as well. Since $B$ is a basis of $\DD$, and the functions $\sp_\DD$ and $\st_\DD$ are Scott continuous we obtain the following consequence:

\begin{lemma}\label{lem-invstsp}
For $x \in D$, $\sp_\DD(\st_\DD(x)) = x$.
\end{lemma}

Thus, both functions are inverse to each other, which shows that $\DD$ is isomorphic to $|\FFF(\DD)|$.

Set $\tau_{\DD} = \st_{\DD}$.

\begin{lemma}\label{lem-tau}
Let $\DD'$ be a further domain with basis $B'$, and $\fun{f}{D}{D'}$ be Scott continuous. 
\begin{enumerate}
\item\label{lem-tau-1}
For $x \in D$, $\DDD(\FFF(f))(x) = \set{a \in B'}{(\exists i \in B) i \ll x \wedge a \ll' f(i)}$.

\item\label{lem-tau-2}
$\tau$ is a natural transformation.

\end{enumerate}
\end{lemma}
\begin{proof}
\eqref{lem-tau-1}
By definition we have for $a \in B'$ that $a \in \DDD(\FFF(f))(x)$, exactly if there are  $i \in B$,  $X \in \Con_{i}$, and $c \in X \cup \{i\}$ so that $X \cup \{i\}  \subseteq \st_{D}(x)$ and $a \ll' f(c)$. As $c \sqsubseteq i$ by the definition of $\Con_{i}$, we have that $a \ll' f(c) \sqsubseteq' f(i)$ whence $a \ll' f(i)$. Conversely, if for some $i \ll x$, $a \ll' f(i)$, set $X = \{i\}$ and $c = i$. Then $a \in \DDD(\FFF(f))(x)$.

\eqref{lem-tau-2}
We have to show that
\begin{equation*}
\tau_{D} \circ \DDD(\FFF(f)) = f \circ \tau_{D'}, 
\end{equation*}
which is a consequence of the continuity of $f$.
\end{proof}

Let us again summarize what we have achieved in this step.

\begin{proposition}\label{pn-tauiso}
$\fun{\tau}{\ID_{\DOM}}{\DDD \circ \FFF}$ is a natural isomorphism.
\end{proposition}

Putting Propositions~\ref{pn-etaiso} and \ref{pn-tauiso} together, we obtain what we were aiming for in this section.

\begin{theorem}\label{thm-eqIFDOM}
 The category $\DOM$ of domains and Scott continuous functions is equivalent to the category $\INF$ of  information frames and approximable mappings.
 \end{theorem}
 
 Because of Lemmas~\ref{lem-funcl} and \ref{lem-funci} we obtain further equivalence results between important subcategories of $\DOM$ and $\INF$, respectively.

\begin{corollary}\label{cor-eqIFDOM} The following categories are equivalent as well:
\begin{enumerate}
\item The category $\DOM$ of domains and Scott continuous functions and the full subcategory $\sINF$ of strong information frames.

\item  The full subcategory $\DOM_{\bot}$ of pointed domains and the subcategories $\INF_{\bt}$ and $\sINF_{\bt}$ of information frames and strong information frames   with truth element and approximable mappings that respect  truth elements, respectively.

\item The full subcategory $\aDOM$  of algebraic domains and the full subcategories $\aINF$ and $\asINF$ of  algebraic information frames  and algebraic strong information frames, respectively.

\item The full subcategory $\aDOM_{\bot}$ of pointed algebraic domains and the subcategories $\aINF_{\bt}$ and $\asINF_{\bt}$ of algebraic and algebraic strong information frames, each of which has truth elements and approximable mappings that respect truth elements.

\end{enumerate}
\end{corollary}

\section{Rough sets}\label{sect-rset}

In this section we study the relationship between information frames and rough sets. We start with introducing the necessary concepts of rough set theory. Rough set theory has been proposed by Z. I. Pawlak~\cite{pa91} as a tool for dealing with the vagueness and granularity in information systems. The core concepts of classical rough sets are lower and upper approximations based on equivalence relations. In Zhu~\cite{zh07} this approach has been generalized by using arbitrary binary relations.

A set $U$ with a binary relation $\Theta$ is called a \emph{generalized approximation space (GA-space, in short)}.

\begin{definition}\label{dn-rs}
Let $(U, \Theta)$ be a GA-space.
\begin{enumerate}

\item\label{dn-rs-1}
For $x \in U$ set $\Theta_{s}(x) = \set{y \in U}{x \Theta y}$.

\item\label{dn-rs-2}
For $X \subseteq U$ define
\begin{equation*}
 \underline{\Theta}(X) = \set{x \in U}{\Theta_{s}(x) \subseteq X} \quad \text{and} \quad \bar{\Theta}(X) = \set{x \in U}{\Theta_{s}(x) \cap X \neq \emptyset}.
 \end{equation*}
 
 \end{enumerate}
\end{definition}
The operators $\fun{\underline{\Theta}, \bar{\Theta}}{\pow{U}}{\pow{U}}$, respectively, are called the \emph{lower  and upper approximation operators} in $\UU$. Such operators are key notions in rough set theory.

\begin{lemma}\label{lem-rsprop}
Let $(U, \Theta)$ be a GA-space. Then  the following statements hold:   
\begin{enumerate}

\item\label{lem-rsprop-1}
$\bar{\Theta}$ is monotone and $\bar{\Theta}(\emptyset) = \emptyset$.

\item\label{lem-rsprop-2}
$\Theta$ is reflexive if, and only if, $X \subseteq \bar{\Theta}(X)$, for all $X \subseteq U$.

\item\label{lem-rsprop-3}
$\Theta$ is transitive if, and only if, $\bar{\Theta}(\bar{\Theta}(X)) \subseteq \bar{\Theta}(X)$, for all $X \subseteq U$.

\item\label{lem-rsprop-4}
Let $\Theta$ be transitive and $X, Y \subseteq U$. Then $\bar{\Theta}(Y) \subseteq \bar{\Theta}(X)$, if $Y \subseteq \bar{\Theta}(X)$. 

\end{enumerate}
\end{lemma}
\begin{proof}
\eqref{lem-rsprop-1} is obvious, as is the left-to-right implication of \eqref{lem-rsprop-2}. For the converse implication set $X = \{x\}$ with $x \in U$. \eqref{lem-rsprop-3} follows similarly and \eqref{lem-rsprop-4} is a consequence of \eqref{lem-rsprop-1} and \eqref{lem-rsprop-3}.
\end{proof}

In \cite{wx23} Wu and Xu introduced a class of GA-spaces that come equipped with a kind of interpolation property.

\begin{definition}\label{dn-cf}
Let $(U, \Theta)$ be a GA-space so that $\Theta$ is transitive. Moreover, let $\FF \subseteq \powf{U}$. 
\begin{enumerate}

\item\label{dn-cf-1}
$\UU = (U, \Theta, \FF)$ is a \emph{generalized approximation space with consistent family of finite subsets}, or in short a \emph{CF-approximation space}, if
\begin{equation}\label{cn-CF}\tag{CF}
(\forall F \in \FF) (\forall K \fsubset \bar{\Theta}(F)) (\exists G \in \FF) K \subseteq \bar{\Theta}(G) \wedge G \subseteq \bar{\Theta}(F).
\end{equation}

\item\label{dn-cf-2}
$\UU$ is a \emph{topological} CF-approximation space if, in addition, $\Theta$ is also reflexive.

\end{enumerate}
\end{definition}

As we will see next, every CF-approximation space generates an information frame. For a CF-approximation space $\UU$ set
\begin{equation*}
\CCC(\UU) = (\FF, (\consc_{F})_{F \in \FF}, (\mmodels_{F})_{F \in \FF})
\end{equation*}
with
\begin{align*}
&\consc_{F} = \{F\} \cup \powf{\set{G \in \FF}{G \subseteq \bar{\Theta}(F)}}, \\
&\XF \mmodels_{F} G \Leftrightarrow (\exists E \in \XF \cup \{F\}) G \subseteq \bar{\Theta}(E).
\end{align*}

\begin{theorem}\label{thm-cfa-inf}
Let $\UU$ be a CF-approximation space. Then $\CCC(\UU)$ is a strong information frame such that
\begin{enumerate}

\item\label{thm-cfa-inf-1}
$\CCC(\UU)$ is algebraic, if $\UU$ is topological.

\item\label{thm-cfa-inf-2}
$\CCC(\UU)$ has a truth element, if the following Condition~\eqref{cn-M} is satified:
\begin{equation}\label{cn-M}\tag{M}
(\exists \bT \in \FF) (\forall F \in \FF) \bT \subseteq \bar{\Theta}(F).
\end{equation}

\end{enumerate}
\end{theorem}
\begin{proof}
We first show that the conditions in Definition~\ref{dn-infofr} are satisfied.
Conditions~\ref{dn-infofr}\eqref{dn-infofr-1} and \ref{dn-infofr}\eqref{dn-infofr-2} obviously hold. In what follows let $F, G \in \FF$.

\ref{dn-infofr}\eqref{dn-infofr-4}
Let $\HF \in \consc_{F}$ and $\YF \fsubset \FF$ with $\HF \mmodels_{F} \YF$. Then, for all $K \in \YF$, there is some $Z_{K} \in \HF \cup \{F\}$ with $K \subseteq \bar{\Theta}(Z_{K})$. Then $Z_{K} \subseteq \bar{\Theta}(F)$ or $Z_{K} = F$. Due to the transitivity of $\Theta$ it follows that $K \subseteq \bar{\Theta}(Z_{K}) \subseteq \bar{\Theta}(\bar{\Theta}(F)) \subseteq \bar{\Theta}(F)$ or $K \subseteq \bar{\Theta}(F)$. Thus, $K \in \set{G \in \FF}{G \subseteq \bar{\Theta}(F)}$, that is, $\YF \in \consc_{F}$.

\ref{dn-infofr}\eqref{dn-infofr-5}
Let  $K \in \FF$  and $\XF, \YF \in \consc_{F}$ with $\XF \subseteq \YF$ and $\XF \mmodels_{F} K$. Then there is some $Z \in \XF \cup \{F\}$ with $K \subseteq \bar{\Theta}(Z)$. Since $\XF \subseteq \YF$, $Z \in \YF \cup \{F\}$ as well. Hence, $\YF \mmodels_{F} K$.

\ref{dn-infofr}\eqref{dn-infofr-6}
Let  $K \in \FF$ and $\XF, \YF \in \consc_{F}$ such that $\XF \mmodels_{F} \YF$ and $\YF \mmodels_{F} K$. Then there is some $Z \in \YF \cup \{F\}$ with $K \subseteq \bar{\Theta}(Z)$. If $Z = F$, then $Z \in \XF \cup \{F\}$ as well and hence we have that $\XF \mmodels_{F} K$.
In the other case that $Z \in \YF$, then, since $\XF \mmodels_{F} \YF$, there is some $T \in \XF \cup \{F\}$ with $Z \subseteq \bar{\Theta}(T)$. Thus, $K \subseteq \bar{\Theta}(Z) \subseteq \bar{\Theta}(\bar{\Theta}(T)) \subseteq \bar{\Theta}(T)$, which means that $\XF \mmodels_{F} K$.

\ref{dn-infofr}\eqref{dn-infofr-7}
Let $\{F\} \in \consc_{G}$ and $\YF \in \consc_{F}$. Then either $\YF = \{F\}$ or for all $K \in \YF$, $K \subseteq \bar{\Theta}(F)$. In the first case, $\YF \in \consc_{G}$, by assumption, and in the second we have that $F \subseteq \bar{\Theta}(G)$. Therefore, $K \subseteq \bar{\Theta}(\bar{\Theta}(G)) \subseteq \bar{\Theta}(G)$ for all $K \in \YF$, which shows that $\YF \in \consc_{G}$ in this case as well.

\ref{dn-infofr}\eqref{dn-infofr-9}
Let $K \in \FF$, $\XF \in \consc_{F}$, and assume that $\{F\} \in \consc_{G}$ and $\XF \mmodels_{F} K$. Then there is some $Z \in \XF \cup \{F\}$ with $K \subseteq \bar{\Theta}(Z)$. If $Z = F$, then $\{Z\} \in \consc_{G}$, by assumption, and hence $Z \subseteq  \bar{\Theta}(G)$. Thus, we have that $K \subseteq \bar{\Theta}(\bar{\Theta}(G)) \subseteq \bar{\Theta}(G)$, which shows that $\XF \mmodels_{G} K$. 
However, if $Z \in \XF$ then $Z \in \XF \cup \{G\}$, which implies that in this case too $\XF \mmodels_{G} K$.

\ref{dn-infofr}\eqref{dn-infofr-11}
Let $\XF \in \consc_{F}$ and $\YF \fsubset \FF$  with $\XF \mmodels_{F} \YF$.  Moreover, let $K \in \YF$. Then there is some $Z \in \XF \cup \{F\}$ with $K \subseteq \bar{\Theta}(Z)$. If $Z = F$, we have that $K \subseteq \bar{\Theta}(F)$. In the other case we have that $Z \subseteq \bar{\Theta}(F)$ and therefore that $K \subseteq \bar{\Theta}(\bar{\Theta}(F)) \subseteq \bar{\Theta}(F)$. It follows that $\bigcup \YF \fsubset \bar{\Theta}(F)$. Since $F \in \FF$, there is thus some $E \in \FF$ with $\bigcup \YF \subseteq \bar{\Theta}(E)$ and $E \subseteq \bar{\Theta}(F)$. Hence, we have that $\XF \mmodels_{F} \{E\}$ and $\{E\} \mmodels_{E} \YF$.

Condition~\eqref{cn-S} is also obviously fulfilled.

For Statement~\eqref{thm-cfa-inf-1} we have to verify Condition~\eqref{cn-A}. Let to this end $\XF \in \consc_{F}$. Note that by Lemma~\ref{lem-rsprop}\eqref{lem-rsprop-2} for all $Z \in \XF \cup \{F\}$, $Z \subseteq \bar{\Theta}(Z)$. Hence, $\XF \mmodels_{F} \{F\} \cup \XF$.

\eqref{thm-cfa-inf-2} 
Let $\bT \in \FF$ with $\bT \subseteq \bar{\Theta}(F)$, for all $F \in \FF$. Then we have that $\emptyset \mmodels_{F} \bT$, for all $F \in \FF$.
\end{proof}
Note that Condition~\eqref{cn-M} is  satisfied in particular, if $\FF$ contains the empty set.

Conversely, also every information frame generates a CF-approximation space. For an information frame $\bA$ define 
\begin{equation*}
\EEE(\bA) = (U, \Theta, \FF)
\end{equation*}
with
\begin{align*}
&U = \bigcup\set{\con_{i} \times \{i\}}{i \in A}, \\
&(X, i) \Theta (Y, j) \Leftrightarrow Y \vdash_{j} \{i\} \cup X, \\
&\FF = \set{\{(X, i)\}}{(X, i) \in U}.
\end{align*}

\begin{theorem}\label{thm-inf-cfa}
Let $\bA$ be an information frame. Then $\EEE(\bA)$ is a CF-approximation space so that
\begin{enumerate}

\item\label{thm-inf-cfa-1}
 $\EEE(\bA)$ is topological if, and only if, $\bA$ is algebraic.
 
 \item\label{thm-inf-cfa-2}
 If $\bA$ has a truth element, then $\EEE(\bA)$  satisfies Condition~\eqref{cn-M}.
 
 \end{enumerate}
\end{theorem}
\begin{proof}
The transitivity of $\Theta$ follows from Lemma~\ref{lem-strong6}.
For Condition~\eqref{cn-CF} let $F \in \FF$, say $F = \{(X, i)\}$, and $K \fsubset \bar{\Theta}(\{(X,i)\})$. Note that $\bar{\Theta}(\{(X, i)\}) = \set{(Y, j) \in U}{X \vdash_{i} \{j\} \cup Y}$. It follows that $X \vdash_{i} \bigcup\set{\{j\} \cup Y}{(Y, j) \in K}$. By Condition~\ref{dn-infofr}\eqref{dn-infofr-11} there is an $e \in A$ and $Z \in \con_{e}$ with $X \vdash_{i} \{e\} \cup Z$ and $Z \vdash_{e} \bigcup\set{\{j\} \cup Y}{(Y, j) \in K}$. Hence, $(Z, e) \Theta(X, i)$, that is, $\{(Z, e)\} \subseteq \bar{\Theta}(\{(X, i)\})$, and $K \subseteq \bar{\Theta}(\{(Z, e)\})$.

For the remaining statements note first that
\begin{equation*}
\text{$\EEE(\bA)$ is topological} \Leftrightarrow
\text{$\Theta$ is reflexive} 
\Leftrightarrow (\forall (X, i) \in U) X \vdash_{i} \{i\} \cup X 
\Leftrightarrow \text{$\bA$ is algebraic},
\end{equation*}
which proves \eqref{thm-inf-cfa-1}. For \eqref{thm-inf-cfa-2}  let $\bt \in A$ be a truth element of $\bA$. Then  $\emptyset \vdash_{i} \bt$, for all $i \in A$, by Condition~\eqref{cn-T}. With Condition~\ref{dn-infofr}\eqref{dn-infofr-5} it follows that for all $X \in \con_{i}$, $X \vdash_{i} \bt$. Thus $\{(\emptyset, \bt)\} \subseteq \bar{\Theta}(\{(X, i)\})$, for all $\{(X, i)\} \in \FF$.
\end{proof}

\section{CF-approximable relations}\label{sect-CFapp}

Wu and Xu~\cite{wx23} also introduced morphisms between CF-approximation spaces.

\begin{definition}\label{dn-apprel}
Let $\UU$ and $\UU'$ be CF-approximation spaces. A relation $\Delta \subseteq \FFF \times \FFF'$ is a \emph{CF-approximable relation from $\UU$ to $\UU'$}, written $\cfrel{\Delta}{\UU}{\UU'}$  if for all $F, F' \in \FF$ and $G, G' \in \FF'$
\begin{enumerate}

\item\label{dn-apprel-1}
$(\exists \hat{G} \in \FF') F \Delta \hat{G}$,

\item\label{dn-apprel-2}
$F \subseteq \bar{\Theta}(F') \wedge F \Delta G \Rightarrow F' \Delta G$,

\item\label{dn-apprel-3}
$F \Delta G \wedge G' \subseteq \bar{\Theta'}(G) \Rightarrow F \Delta G'$,

\item\label{dn-apprel-4}
$F \Delta G \Rightarrow (\exists \hat{F} \in \FF) (\exists \hat{G} \in \FF') \hat{F} \subseteq \bar{\Theta}(F) \wedge G \subseteq \bar{\Theta'}(\hat{G}) \wedge \hat{F} \Delta \hat{G}$,

\item\label{dn-apprel-5}
$F \Delta G \wedge F \Delta G' \Rightarrow (\exists \hat{G} \in \FF') G \cup G' \subseteq \bar{\Theta'}(\hat{G}) \wedge F \Delta \hat{G}$.

\end{enumerate}
\end{definition}

For a CF-approximation space $\UU$ let $\Id_{\UU} \subseteq \FF \times \FF$ be given by 
\begin{equation*}
F \Id_{\UU} G \Leftrightarrow G \subseteq \bar{\Theta}(F).
\end{equation*}
Moreover, for CF-approximation spaces $\UU^{(\nu)}$ with $\nu = 1,2,3$,  $\cfrel{\Delta}{\UU^{(1)}}{\UU^{(2)}}$ and $\cfrel{\Omega}{\UU^{(2)}}{\UU^{(3)}}$, define the relation  $\Delta \circ \Omega \subseteq \FF^{(1)} \times \FF^{(3)}$ by
\begin{equation*}
F (\Delta \circ \Omega) G \Leftrightarrow (\exists E \in \FF^{(2)}) F \Delta E \wedge  E \Omega G.
\end{equation*}

\begin{lemma}\label{lem-cfrelprop}
For $\nu = 1, 2, 3$, let $\UU^{(\nu)}$ be CF-approximation spaces, $\cfrel{G}{\UU^{(1)}}{\UU^{(2)}}$ and $\cfrel{H}{\UU^{(2)}}{\UU^{(3)}}$. Then the following statements hold:
\begin{enumerate}

\item\label{lem-cfrelprop-1}
 $\cfrel{\Id_{\UU^{(\nu)}}}{\UU^{(\nu)}}{\UU^{(\nu)}}$.

\item\label{lem-cfrelprop-2}
$\cfrel{\Delta \circ \Omega}{\UU^{(1)}}{\UU^{(3)}}$. 

\item\label{lem-cfrelprop-4}
$\Id_{\UU^{(1)}} \circ \Delta = \Delta \circ \Id_{\UU^{(2)}} = \Delta$.

\end{enumerate}
\end{lemma}

In the following let $\CFA$ be the category of CF-approximation spaces and CF-approximable relation, and let $\tCFA$, $\CFAM$ and $\tCFAM$ be respectively: the full subcategory of topological CF-approximation spaces, those that satisfy Condition~\eqref{cn-M}, and those that do both, are topological and satisfy Condition~\eqref{cn-M}. 

By the preceding lemma $\Id_{\UU}$ is the identity morphism on CF-approximation space $\UU$. We have already shown that CF-approximation spaces lead to information frames and vice versa. We will now do the same for CF-approximable relations and approximable mappings.

Let  $\UU$, $\UU'$ be CF-approximation spaces and $\Delta$ a CF-approximable relation from $\UU$ to $\UU'$. 
Set 
\begin{equation*}
\CCC(\Delta) = H_{\Delta} = (H^{\Delta}_{F})_{F\in \FF},
\end{equation*}
where for $F \in \FF$, $\XF \in \consc_{F}$ and $G \in \FF'$, 
\begin{equation*}
\XF H^{\Delta}_{F} G \Leftrightarrow (\exists Z \in \XF \cup \{F\}) Z \Delta G.
\end{equation*}

\begin{proposition}\label{pn-ar-am}
Let  $\UU$, $\UU'$ be CF-approximation spaces and $\cfrel{\Delta}{\UU}{\UU'}$. Then 
\begin{equation*}
\appmap{H_{\Delta}}{\CCC(\UU)}{\CCC(\UU')}.
\end{equation*}
 Moreover, if $\UU$ and $\UU'$ both satisfy Condition~\eqref{cn-M} then $H_{\Delta}$ respects truth elements.
\end{proposition}
\begin{proof}
We have to verify Conditions~\ref{dn-am}\eqref{dn-am-1}-\eqref{dn-am-5}. Let  to this end $F, F' \in \FF$, $\XF, \XF' \in \consc_{F}$, $E, K \in \FF'$, and $\YF \in \consc'_{K}$.

\eqref{dn-am-1}
Assume that $\XF H^{\Delta}_{F} (\{K\} \cup \YF)$ and $\YF \mmodels_{K} E$. We have to show that $\XF H^{\Delta}_{F} E$. By our assumption we have that
\begin{equation*}
(\forall G \in \YF \cup \{K\}) (\exists Z_{G} \in \XF \cup \{F\}) Z_{G} \Delta G \quad\text{and}\quad
(\exists T \in \YF \cup \{K\}) E \subseteq \bar{\Theta'}(T). 
\end{equation*} 
It follows that $Z_{T} \Delta T$ and $E \subseteq \bar{\Theta'}(T)$. With Condition~\ref{dn-apprel}\eqref{dn-apprel-3} we thus obtain that $Z_{T} \Delta E$, that is, $\XF H^{\Delta}_{F} E$.

\eqref{dn-am-2}
Suppose that $\XF \subseteq \XF'$ and $\XF H^{\Delta}_{F} E$. We have to show that $\XF' H^{\Delta}_{F} E$. By our assumption we have that there is some $Z \in \XF \cup \{F\}$ with $Z \Delta E$. Then $Z \in \XF' \cup \{F\}$ too and therefore $\XF' H^{\Delta}_{F} E$.

\eqref{dn-am-3}
Assume that $\XF \mmodels_{F} \XF'$ and $\XF' H^{\Delta}_{F} E$. We must show that $\XF H^{\Delta}_{F} E$. By our supposition we have that there is some $Z \in \XF' \cup \{F\}$ with $Z \Delta E$. If $Z = F$ it follows that $\XF H^{\Delta}_{F} E$. Otherwise, if $Z \in \XF'$, there is some $T \in \XF \cup \{F\}$ with $Z \subseteq \bar{\Theta}(T)$. With Condition~\ref{dn-apprel}\eqref{dn-apprel-2} we therefore obtain that $T \Delta E$. Thus, $\XF H^{\Delta}_{F} E$.

\eqref{dn-am-4}
Let $\{F\} \in \consc_{F'}$ and $\XF H^{\Delta}_{F} E$. By \ref{dn-infofr}\eqref{dn-infofr-7} and the first assumption, $\XF \in \consc_{F'}$. Moreover, by the second one, there is some $Z \in \XF \cup \{F\}$ with $Z \Delta E$. If $Z \in \XF$, then we have that  $\XF H^{\Delta}_{F'} E$. In the other case that $Z = F$, we know that $\{F\} \in \consc_{F'}$, and hence that $F \subseteq \bar{\Theta}(F')$. With Property~\ref{dn-apprel}\eqref{dn-apprel-2} it follows that also $F' \Delta E$. Thus, there is some $Z' \in \XF \cup \{F'\}$ with $Z' \Delta E$, that is, $\XF H^{\Delta}_{F'} E$.

\eqref{dn-am-5}
Let $\XF H^{\Delta}_{F} \LF$ with $\LF \fsubset \FF'$. Then there is some $Z_{L} \in \XF \cup \{F\}$ with $Z_{L} \Delta L$, for each $L \in \LF$. With Condition~\ref{dn-apprel}\eqref{dn-apprel-4} we furthermore obtain, that for each $L \in \LF$ there are $M_{L} \in \FF$ and $N_{L} \in \FF'$ so that $M_{L} \subseteq \bar{\Theta}(Z_{L})$, $L \subseteq \bar{\Theta'}(N_{L})$ and $M_{L} \Delta N_{L}$.

Since $M_{L} \subseteq \bar{\Theta}(Z_{L})$ and $Z_{L} \in \XF \cup \{F\}$, it follows that $M_{L} \subseteq \bar{\Theta}(F)$. Set $\UF = \set{M_{L}}{L \in \LF}$. Because of Property~\eqref{cn-CF} there is thus some $C \in \FF$ with $\bigcup \UF \subseteq \bar{\Theta}(C)$ and $C \subseteq \bar{\Theta}(F)$. Thus, $\XF \mmodels_{F} C$ and $\{C\} \mmodels_{C} \UF$. As $M_{L} \subseteq \bar{\Theta}(Z_{L})$, we moreover have that $\XF \mmodels_{F} \UF$. In addition, since $\{C\} \mmodels_{C} \UF$, $\UF \in \consc_{C}$.

Set $\VF = \set{N_{L}}{L \in \LF}$. Because, for $L \in \LF$, $M_{L} \Delta N_{L}$ and $M_{L} \in \UF$, it follows that $\UF H^{\Delta}_{C} N_{L}$. Hence, $\UF H^{\Delta}_{C} \VF$.

As shown above, $M_{L} \subseteq \bar{\Theta}(C)$, for all $L \in \LF$. As $M_{L} \Delta N_{L}$, it follows with Property~\ref{dn-apprel}\eqref{dn-apprel-2} that $C \Delta N_{L}$, for all $L \in \LF$. Note that $\LF$ is finite. Therefore, by Property~\ref{dn-apprel}\eqref{dn-apprel-5}, there is some $Q \in \FF'$ with $\bigcup \VF \subseteq \bar{\Theta'}(Q)$ and $C \Delta Q$. Thus, we have that $\UF H^{\Delta}_{C} Q$. Moreover, $\VF \in \consc'_{Q}$
Finally, as $L \subseteq \bar{\Theta'}(N_{L})$, for all $L \in \LF$, it also follows that $\VF \mmodels_{Q} \LF$.

For the remaining statement assume that $\UU$ and $\UU'$ have Property~\eqref{cn-M}. Then there are $\bT \in \FF$ and $\bT' \in \FF'$ with $\bT \subseteq \bar{\Theta}(F)$ and $\bT' \subseteq \bar{\Theta'}(G)$, for all $F \in \FF$ and $G \in \FF'$. By \ref{dn-apprel}\eqref{dn-apprel-1} there is some $G \in \FF'$ with $\bT \Delta G$. Hence also $\bT \Delta \bT'$, by \ref{dn-apprel}\eqref{dn-apprel-3}.  It follows that $\emptyset H^{\Delta}_{\bT} \bT'$. 
\end{proof}

\begin{lemma}\label{lem-funcC}
$\fun{\CCC}{\CFA}{\sINF}$ is a functor such that 
\begin{enumerate}

\item
$\CCC[\tCFA] \subseteq \asINF$, 

\item
$\CCC[\CFAM] \subseteq \sINF_{\bt}$, and 

\item
$\CCC[\tCFAM] \subseteq \asINF_{\bt}$.

\end{enumerate}
\end{lemma}

Next, we consider the converse situation. As will be shown, every approximable mapping  between information frames generates a CF-approximable relation between the derived CF-approximation spaces. Let $\bA$, $\bA'$ be information frames and $\appmap{H = (H_{i})_{i \in A}}{\bA}{\bA'}$ an approximable mapping. Then for  $\{(X, i)\} \in \FF$ and $\{(Y, j)\} \in \FF'$ define
\begin{equation*}
\{(X, i)\} \Delta_{H} \{(Y, j)\} \Leftrightarrow X H_{i} (\{j\} \cup Y).
\end{equation*}

\begin{proposition}\label{pn-am-ar}
Let  $\bA$, $\bA'$ be information frames and $\appmap{H = (H_{i})_{i \in A}}{\bA}{\bA'}$. Then 
\begin{equation*}
\cfrel{\Delta_{H}}{\EEE(\bA)}{\EEE(\bA')}.
\end{equation*} 
Moreover, if $\bA$ and $\bA'$ have truth elements $\bt$ and $\bt'$, respectively, which are respected  by $H$, then $\{(\emptyset, \bt)\} \Delta_{H} \{(\emptyset, \bt')\}$.
\end{proposition}
\begin{proof}
\eqref{dn-apprel-1}
Let $\{(X, i)\} \in \FF$ and  $K = \emptyset$. Then $X H_{i} K$ holds. By Lemma~\ref{pn-amint}\eqref{pn-amint-2} there are  hence $e \in A'$ and $V \in \con'_{e}$ such that $X H_{i} (\{e\} \cup V)$ and $V \vdash'_{e} K$. Since $V \in \con'_{e}$, we have that $\{(V, e)\} \in \FF'$. This shows that there is some $\{(V, e)\} \in \FF'$ with $\{(X, i)\} \Delta_{H} \{(V, e)\}$.

\eqref{dn-apprel-2}
Let $\{(X, i)\}, \{(X', i')\} \in \FF$ and $\{(Y, j)\} \in \FF'$ so that $\{(X, i)\} \subseteq \bar{\Theta}(\{(X', i')\})$ and $\{(X, i)\} \Delta_{H} \{(Y, j)\}$. Then $X' \vdash_{i'} \{i\} \cup X$ and $X H_{i} (\{j\} \cup Y)$.  With \ref{dn-am}\eqref{dn-am-3}, \eqref{dn-am-4} it follows that $X' H_{i} (\{j\} \cup Y)$, that is, $\{(X', i')\} \Delta_{H} \{(Y, j)\}$.

\eqref{dn-apprel-3}
Let $\{(X, i)\} \in \FF$ and $\{(Y, j)\}, \{(Y', j')\} \in \FF'$ so that we have $\{(X, i)\} \Delta_{H} \{(Y, j)\}$ and $\{(Y', j')\} \subseteq \bar{\Theta'}(\{(Y, j)\})$. Then $X H_{i} (\{j\} \cup Y)$ and $Y \vdash'_{j} (\{j'\} \cup Y')$. By applying  \ref{dn-am}\eqref{dn-am-1} we obtain that $X H_{i} (\{j'\} \cup Y')$, that is, $\{(X, i)\} \Delta_{H} \{(Y', j')\}$.

\eqref{dn-apprel-4}
Let $\{(X, i)\} \in \FF$ and $\{(Y, j)\} \in \FF'$ with $\{(X, i)\} \Delta_{H} \{(Y, j)\}$. Then $X H_{i} (\{j\} \cup Y)$. By \ref{dn-am}\eqref{dn-am-5} there are $c \in A$, $U \in \con_{c}$, $e \in A'$ and $V \in \con'_{e}$ so that $X \vdash_{i} \{c\} \cup U$, $U H_{c} (\{e\} \cup V)$ and $V \vdash'_{e} \{j\} \cup Y$. Thus, $\{(U, c)\} \subseteq \bar{\Theta}(\{(X, i)\})$, $\{(U, c)\} \Delta_{H} \{(V, e)\}$ and $\{(Y, j)\} \subseteq \bar{\Theta'}(\{(V, e)\})$.

\eqref{dn-apprel-5}
Let $\{(X, i)\} \in \FF$ and for $\nu = 1, 2$ $\{(Y_{\nu}, j_{\nu})\} \in \FF'$ with $\{(X, i)\} \Delta_{H} \{(Y_{\nu}, j_{\nu})\}$. Then $X H_{i} (\{j_{1}, j_{2}\} \cup Y_{1} \cup Y_{2})$. By Lemma~\ref{pn-amint}\eqref{pn-amint-2} there is thus some $e \in A'$ and $V \in \con'_{e}$ with $X H_{i} (\{e\} \cup V)$ and $V \vdash'_{e} \{j_{1}, j_{2}\} \cup Y_{1} \cup Y_{2}$. Hence, $\{(X, i)\} \Delta_{H} \{(V, e)\}$ and $\{(Y_{1}, j_{1}), (Y_{2}, j_{2})\} \subseteq \bar{\Theta'}(\{(V, e)\})$.

The remaining part of the assertion follows immediately.
\end{proof}

Set $\EEE(H) = \Delta_{H}$.

\begin{lemma}\label{lem-funcIC}
$\fun{\EEE}{\INF}{\CFA}$ is a functor such that
\begin{enumerate}

\item
$\EEE[\INF_{\bt}] \subseteq \CFAM$,

\item
$\EEE[\aINF] \subseteq \tCFA$,

\item 
$\EEE[\aINF_{\bt}] \subseteq \tCFAM$.

\end{enumerate}
\end{lemma}

As we will show next, the functors $\fun{\CCC}{\CFA}{\INF}$ and $\fun{\EEE}{\INF}{\CFA}$ establish an equivalence between the categories $\CFA$  and $\INF$. We first construct a natural isomorphism $\fun{\delta}{\Id_{\CFA}}{\EEE \circ \CCC}$. Let to this end $\UU$ be a CF-approximation space. Then 
\begin{equation*}
\EEE(\CCC(\UU)) = (\tilde{U}, \tilde{\Theta}, \tilde{\FF})
\end{equation*}
with
\begin{align*}
&\tilde{U} = \bigcup\set{\consc_{F} \times \{F\}}{F \in \FF}, \\
&(\XF, F) \tilde{\Theta} (\YF, G) \Leftrightarrow (\forall K \in \{F\} \cup \XF) (\exists Z_{K} \in  \{G\} \cup \YF) K \subseteq \bar{\Theta}(Z_{K}), \\
&\tilde{\FF} = \set{\{(\XF, F)\}}{(\XF, F) \in \tilde{U}}.
\end{align*}

Define relations $\Upsilon_{\UU} \subseteq \FF \times \tilde{\FF}$ and $\Gamma_{\UU} \subseteq  \tilde{\FF} \times \FF$ by
\begin{align*}
&F \Upsilon_{\UU} \{(\YF, G)\} \Leftrightarrow (\forall K \in \{G\} \cup \YF) K \subseteq \bar{\Theta}(F), \\
&\{(\XF, F)\} \Gamma_{\UU} G \Leftrightarrow (\exists K \in \{F\} \cup \XF) G \subseteq \bar{\Theta}(K).
\end{align*}

\begin{lemma}\label{lem-C}
\begin{enumerate}

\item\label{lem-C-1}
$\cfrel{\Upsilon_{\UU}}{\FF}{\tilde{\FF}}$.

\item\label{lem-C-2}
$\cfrel{\Gamma_{\UU}}{\tilde{\FF}}{\FF}$.

\item\label{lem-C-3}
$\Upsilon_{\UU} \circ \Gamma_{\UU} = \Id_{\UU}$.

\item\label{lem-C-4}
$\Gamma_{\UU} \circ \Upsilon_{\UU} = \Id_{\EEE(\CCC(\UU))}$.

\end{enumerate}
\end{lemma}
\begin{proof}
\eqref{lem-C-1}
We have to verify the conditions in Definition~\ref{dn-apprel}:

\ref{dn-apprel}\eqref{dn-apprel-1}
Let $F \in \FF$ and apply Condition~\eqref{cn-CF} for $K = \emptyset$. Then there is some $G \in \FF$ with $G \subseteq \bar{\Theta}(F)$. Set $\YF = \{G\}$. Then $\YF \in \consc_{G}$ and hence $\{(\YF, G)\} \in \tilde{\FF}$. It follows for $K' \in \YF \cup \{G\}$ that $K' = G \subseteq \bar{\Theta}(F)$. Thus, $F \Upsilon_{\UU} \{(\YF, G)\}$.

\ref{dn-apprel}\eqref{dn-apprel-2}
Let $F, F' \in \FF$ and $\{(\YF, G)\} \in \tilde{\FF}$ with $F \subseteq \bar{\Theta}(F')$ and $F \Upsilon_{\UU} \{(\YF, G)\}$. Then we have for all $K \in \YF \cup \{G\}$ that $K \subseteq \bar{\Theta}(F)$. Since $F \subseteq \bar{\Theta}(F')$, it follows by the transitivity of $\Theta$ that $K \fsubset \bar{\Theta}(F')$. Hence, $F' \Upsilon_{\UU} \{(\YF, G)\}$. 

\ref{dn-apprel}\eqref{dn-apprel-3}
Let $F \in \FF$ and $\{(\YF, G)\}, \{(\YF', G')\} \in \tilde{\FF}$ with 
\begin{equation*}
F \Upsilon_{\UU} \{(\YF, G)\} \quad\text{and}\quad \{(\YF', G')\} \subseteq  \bar{\tilde{\Theta}}(\{(\YF, G)\}).
\end{equation*}
Then we have that $K \subseteq \bar{\Theta}(F)$ for all $K \in \YF \cup \{G\}$ and that, for each $K' \subseteq \YF' \cup \{G\}$. $L \in \YF \cup \{G\}$ with $K' \fsubset \bar{\Theta}(L)$. Hence, $K' \subseteq \bar{\Theta}(L) \subseteq \bar{\Theta}(F)$ in view of transitivity and Condition~\ref{dn-apprel}\eqref{dn-apprel-3}.

\ref{dn-apprel}\eqref{dn-apprel-4}
Let $F \in \FF$ and $\{(\YF, G)\} \in \tilde{\FF}$ with $F \Upsilon_{\UU} \{(\YF, G)\}$. Then we have for all $K \in \YF \cup \{G\}$ that $K \subseteq \bar{\Theta}(F)$. Let $M \in \YF \cup \{G\}$. Then, by Condition~\eqref{cn-CF}, there is some $F_{M} \in \FF$ so that
\begin{align}
&M \subseteq \bar{\Theta}(F_{M}) \quad\text{and} \label{eq-stern} \\
&F_{M} \subseteq \bar{\Theta}(F). \label{eq-2stern} 
\end{align}
As a consequence of \eqref{eq-2stern} we obtain that $\bigcup\set{F_{M}}{M \in \YF \cup \{G\}} \subseteq \bar{\Theta}(F)$. Since $\YF$ is finite, it follows with Property~\eqref{cn-CF} that there is some $F' \in \FF$ so that 
\begin{equation}\label{eq-3stern}
\bigcup\set{F_{M}}{M \in \YF \cup \{G\}} \subseteq \bar{\Theta}(F')
\end{equation}
and $F' \subseteq \bar{\Theta}(F)$.

If $\YF = \{G\}$, set $G' = F_{G}$ and $\YF' = \{G'\}$. Then $\YF' \in \consc_{G'}$ and hence $\{(\YF', G')\} \in \tilde{\FF}$. Because of \eqref{eq-stern} we moreover have that $G \subseteq \bar{\Theta}(G')$. Thus, $\{(\YF, G)\} \subseteq \bar{\tilde{\Theta}}(\{(\YF', G')\})$. Finally,  it follows  by \eqref{eq-3stern} that $G' \subseteq \bar{\Theta}(F')$, that is, $F' \Upsilon_{\UU} \{(\YF', G')\}$.

If, on the other hand, $\YF \neq \{G\}$, we have for all $K \in \YF$ that $K \subseteq \bar{\Theta}(G)$. As by \eqref{eq-stern} $G \subseteq \bar{\Theta}(F_{G})$, it follows that $K \subseteq \bar{\Theta}(F_{G})$. Set $G' = F_{G}$ and $\YF' = \{G'\}$. Then $\YF' \in \consc_{G'}$, that is, $(\YF', G') \in \tilde{\FF}$, and $(\YF, G) \tilde{\Theta} (\YF', G')$. With \eqref{eq-3stern} we moreover have that $G' \subseteq \bar{\Theta}(F')$, that is, $F' \Upsilon_{\UU} \{(\YF', G')\}$.

\ref{dn-apprel}\eqref{dn-apprel-5}
Let $F \in \FF$ and for $\nu = 1, 2$, $\{(\YF_{\nu}, G_{\nu})\} \in \tilde{\FF}$ with $F \Upsilon_{\UU} \{(\YF_{\nu}, G_{\nu})\}$. Then we have for all $K \in \YF_{1} \cup \YF_{2} \cup \{G_{1}, G_{2}\}$ that $K \subseteq \bar{\Theta}(F)$. It follows that $\bigcup (\YF_{1} \cup \YF_{2} \cup \{G_{1}, G_{2}\}) \fsubset \bar{\Theta}(F)$. Thus, there is some $Z \in \FF$ with $\bigcup (\YF_{1} \cup \YF_{2} \cup \{G_{1}, G_{2}\}) \subseteq \bar{\Theta}(Z)$ and 
\begin{equation}\label{eq-circ}
Z \subseteq \bar{\Theta}(F).
\end{equation}
Set $\ZF = \{Z\}$. Then $\ZF \in \consc_{Z}$ and hence $\{(\ZF, Z)\} \in \tilde{\FF}$. Moreover, we have for all $K \in \YF_{1} \cup \YF_{2} \cup \{G_{1}, G_{2}\}$ that $K \subseteq \bar{\Theta}(Z)$. Thus, $\{(\YF_{1}, G_{1})\} \cup \{(\YF_{2}, G_{2})\} \subseteq \bar{\tilde{\Theta}}(\{(\ZF, Z)\})$. Because of \eqref{eq-circ} we moreover have that $F \Upsilon_{\UU} \{(\ZF, Z)\}$.

\eqref{lem-C-2}
Again we have to verify Conditions~\ref{dn-apprel}\eqref{dn-apprel-1}-\eqref{dn-apprel-5}:

\ref{dn-apprel}\eqref{dn-apprel-1}
Let $\{(\XF, F)\} \in \tilde{\FF}$. Dann ist $\XF \in \consc_{F}$, that is, either $\XF = \{F\}$ or $K \subseteq \bar{\Theta}(F)$, for all $K \in \XF$.

If $\XF = \{F\}$, then by Condition~\eqref{cn-CF}, as $\emptyset \fsubset \bar{\Theta}(F)$, there is some $G \in \FF$ with $\emptyset \subseteq \bar{\Theta}(G)$ and $G \subseteq \bar{\Theta}(F)$. Thus, $\{(\XF, F)\} \Gamma_{\UU} G$.

On the other hand, if $\XF \neq \{F\}$, then we have for all $K \in \XF$ that $K \subseteq \bar{\Theta}(F)$. Therefore, $\bigcup \XF \fsubset \bar{\Theta}(F)$.  By Condition~\eqref{cn-CF}, again, there is some $G \in \FF$ with $\bigcup \XF \fsubset \bar{\Theta}(G)$ and $G \subseteq \bar{\Theta}(F)$. It follows that $\{(\XF, F)\} \Gamma_{\UU} G$.

\ref{dn-apprel}\eqref{dn-apprel-2}
Let $\{(\XF, F)\}, \{(\XF', F')\} \in \tilde{\FF}$ and $G \in \tilde{\FF'}$ with 
\begin{equation*}
\{(\XF, F)\} \subseteq \bar{\tilde{\Theta}}(\{(\XF', F')\}) \quad\text{and}\quad \{(\XF, F)\} \Gamma_{\UU} G. 
\end{equation*}
Then we have
\begin{equation*}
(\forall K \in \XF \cup \{F\}) (\exists L \in \XF' \cup \{F'\}) K \subseteq \bar{\Theta}(L) \quad\text{and} \quad
(\exists M \in \XF \cup \{F\}) G \subseteq \bar{\Theta}(M).
\end{equation*}
We have to show that for some $M' \in \XF' \cup \{F'\}$, $G \subseteq \bar{\Theta}(M')$. Let $M \in \XF \cup \{F\}$ with $G \subseteq \bar{\Theta}(M)$. Then there exists $M' \in \XF' \cup \{F'\}$ such that $M \subseteq \bar{\Theta}(M')$. It follows that $G \subseteq \bar{\Theta}(M')$.

\ref{dn-apprel}\eqref{dn-apprel-3}
Let $\{(\XF, F)\} \in \tilde{\FF}$ and $G, G' \in \FF$ so that $\{(\XF, F)\} \Gamma_{\UU} G$ and $G' \subseteq \bar{\Theta}(G)$. Then there is some $K \in \XF \cup \{F\}$ with $G \subseteq \bar{\Theta}(K)$. It follows that also $G' \subseteq \bar{\Theta}(K)$, that is, $\{(\XF, F)\} \Gamma_{\UU} G'$.

\ref{dn-apprel}\eqref{dn-apprel-4}
Let $\{(\XF, F)\} \in \tilde{\FF}$ and $G \in \FF$ such that $\{(\XF, F)\} \Gamma_{\UU} G$. Then there exists $K \in \XF \cup \{F\}$ with $G \subseteq \bar{\Theta}(K)$. Thus, there is some $K' \in \FF$ so that $G \subseteq \bar{\Theta}(K')$ and $K' \subseteq \bar{\Theta}(K)$. Set $F' = K'$ and $\XF' = \{F'\}$. Then $(\XF', F') \tilde{\Theta} (\XF, F)$, that is, $\{(\XF', F')\} \subseteq \bar{\tilde{\Theta}}(\{(\XF, F)\})$. Moreover, as $G \subseteq \bar{\Theta}(K')$, there is some $K'' \in \FF$ with $G \subseteq \bar{\Theta}(K'')$ and $K'' \subseteq \bar{\Theta}(K')$. Set $G' = K''$. Then $G \subseteq \bar{\Theta}(G')$. In addition, $G' \subseteq \bar{\Theta}(F')$, which means that $\{(\XF', F')\} \Gamma_{\UU} G'$.

\ref{dn-apprel}\eqref{dn-apprel-5}
Let $\{(\XF, F)\} \in \tilde{\FF}$ and for $\nu = 1, 2$, $G_{\nu} \in \FF$ with $\{(\XF, F)\} \Gamma_{\UU} G_{\nu}$. Then there exists $K_{\nu} \in \XF \cup \{F\}$ so that $G_{\nu} \subseteq \bar{\Theta}(K_{\nu})$. 

If $\XF = \{F\}$, we have that $K_{1} = K_{2} = F$. Thus $G_{1} \cup G_{2} \subseteq \bar{\Theta}(F)$. Therefore, there is some $\hat{G} \in \FF$ with $G_{1} \cup G_{2} \subseteq \bar{\Theta}(\hat{G})$ and $\hat{G} \subseteq \bar{\Theta}(F)$, where the latter property means that $\{(\XF, F)\} \Gamma_{\UU} \hat{G}$.

If, however,  $\XF \neq \{F\}$, then we have for all $M \in \XF$ that $M \subseteq \bar{\Theta}(F)$. Moreover, there is some $M_{\nu} \in \XF$, for $\nu =1, 2$, with $G_{\nu} \subseteq \bar{\Theta}(M_{\nu}) \subseteq \bar{\Theta}(F)$. Thus, $G_{1} \cup G_{2} \subseteq \bar{\Theta}(F)$. Now, it follows as in the preceding case that there is some $\hat{G} \in \FF$ with $G_{1} \cup G_{2} \subseteq \bar{\Theta}(\hat{G})$ and $\hat{G} \subseteq \bar{\Theta}(F)$. As $F \in \XF \cup \{F\}$, the latter property implies that $\{(\XF, F)\} \Gamma_{\UU} \hat{G}$.

\eqref{lem-C-3}
Let $F, G \in \FF$. Then we have that
\begin{align*}
F (\Upsilon_{\UU} \circ \Gamma_{\UU}) G 
&\Leftrightarrow (\exists \{(\XF, L)\} \in \tilde{\FF}) F \Upsilon_{\UU} \{(\XF, L)\} \wedge \{(\XF, L)\} \Gamma_{\UU} G \\
&\Leftrightarrow (\exists \{(\XF, L)\} \in \tilde{\FF}) (\forall K \in \XF \cup \{L\}) K \subseteq \bar{\Theta}(F) \wedge (\exists M \in \XF \cup \{L\}) G \subseteq \bar{\Theta}(M) \\
&\Rightarrow G \subseteq \bar{\Theta}(F) \\
&\Leftrightarrow F \Id_{\UU} G.
\end{align*}

Conversely, if $G \subseteq \bar{\Theta}(F)$, then there is some $\hat{G} \in \FF$ with $G \subseteq \bar{\Theta}(\hat{G})$ and $\hat{G} \subseteq \bar{\Theta}(F)$. Set $\XF = \{\hat{G}\}$. Then $\XF \in \consc_{\hat{G}}$ and hence $(\XF, \hat{G}) \in \FF$. Moreover, $F \Upsilon_{\UU} \{(\XF, \hat{G})\}$ and $\{(\XF, \hat{G})\} \Gamma_{\UU} G$. Hence we have that $F (\Upsilon_{\UU} \circ \Gamma_{\UU}) G$.

\eqref{lem-C-4}
Let $\{(\XF, F)\}, \{(\YF, G)\} \in \tilde{\FF}$. Then we have that
\begin{align*}
\{(\XF, F)\} (\Gamma_{\UU}& \circ \Upsilon_{\UU}) \{(\YF, G)\} \\
&\Leftrightarrow (\exists L \in \FF) \{(\XF, F)\} \Gamma_{\UU} L \wedge L \Upsilon_{\UU} \{(\YF, G)\} \\
&\Leftrightarrow (\exists L \in \FF)(\exists K \in \XF \cup \{F\}) L \subseteq \bar{\Theta}(K) \wedge (\forall M \in \YF \cup \{G\}) M \subseteq \bar{\Theta}(L) \\
&\Rightarrow (\forall M \in \YF \cup \{G\}) (\exists K \in \XF \cup \{F\}) M \subseteq \bar{\Theta}(K) \\
&\Rightarrow (\YF, G) \tilde{\Theta} (\XF, F) \\
&\Leftrightarrow \{(\YF, G)\} \subseteq \bar{\tilde{\Theta}}(\{(\XF, F)\}) \\
&\Leftrightarrow \{(\XF, F)\} \Id_{\EEE(\CCC(\UU))} (\{(\YF, G)\}).
\end{align*}

Conversely, let $\{(\YF, G)\} \subseteq \bar{\tilde{\Theta}}(\{(\XF, F)\}$. Then there is some $\{(\ZF, Z)\} \in \tilde{\FF}$ with $\{(\YF, G)\}  \subseteq \bar{\tilde{\Theta}}(\{(\ZF, Z)\})$ and $\{(\ZF, Z)\} \subseteq \bar{\tilde{\Theta}}(\{(\XF, F)\})$. It follows that
\begin{align}
&(\forall M \in \YF \cup \{G\}) (\exists L \in \ZF \cup \{Z\}) M \subseteq \bar{\Theta}(L), \label{cn-nstern} \quad\text{and} \\
&(\forall L' \in \ZF \cup \{Z\}) (\exists N_{L'} \in \XF \cup \{F\}) L' \subseteq \bar{\Theta}(N_{L'}). \label{cn-2nstern}
\end{align}
As $\ZF \in \consc_{Z}$ we have that either $\ZF = \{Z\}$, or $\ZF \neq \{Z\}$ and $K \subseteq \bar{\Theta}(Z)$, for all $K \in \ZF$. Therefore, \eqref{cn-nstern} implies 
 that $M \subseteq \bar{\Theta}(Z)$, for all $M \in \YF \cup \{G\}$, that is, we have that $Z \Upsilon_{\UU} \{(\YF, G)\}$. Next, choose $L' = Z$ in \eqref{cn-2nstern}. Then we obtain in particular that $N_{Z} \in \XF \cup \{F\}$ with $Z \subseteq \bar{\Theta}(N_{Z})$, that is, we have $\{(\XF, F)\} \Gamma_{\UU} Z$. This shows that $\{(\XF, F)\} (\Gamma_{\UU} \circ \Upsilon_{\UU}) \{(\YF, G)\}$.
\end{proof}

Set $\delta_{\UU} = \Upsilon_{\UU}$.

\begin{lemma}\label{lem-natrdelt}
$\fun{\delta}{\ID_{\CFA}}{\EEE \circ \CCC}$ is a natural transformation. 
\end{lemma}
\begin{proof}
Let $\UU, \UU'$ be CF-approximation spaces,  and $\cfrel{\Delta}{\UU}{\UU'}$ a CF-approximable relation from $\UU$ to $\UU'$. Set $\tilde {\Delta} = \EEE(\CCC(\Delta))$. Then, for $\{(\XF, F)\} \in \tilde{\FF}$ and $\{(\YF, G)\} \in \tilde{\FF'}$,
\begin{equation*}
\{(\XF, F)\} \tilde{\Delta} \{(\YF, G)\} \Leftrightarrow (\forall K \in \YF \cup \{G\}) (\exists L \in \XF \cup \{F\}) L \Delta K,
\end{equation*}

We have to show that $\delta_{\UU} \circ \tilde{\Delta} = \Delta \circ \delta_{\UU'}$. Let to this end $F \in \FF$ and $\{(\XF', F')\} \in \tilde{\FF'}$. Then we have
\begin{equation}\label{eq-rh}
F (\delta_{\UU} \circ \tilde{\Delta}) \{(\XF', F')\} \Leftrightarrow (\exists (\ZF, Z) \in \tilde{\FF}) F \Upsilon_{\UU} \{(\ZF, Z)\} \wedge \{(\ZF, Z)\} \tilde{\Delta} \{(\XF', F')\}.
\end{equation}
By definition 
\begin{equation*}
F \Upsilon_{\UU} \{(\ZF, Z)\} 
\Leftrightarrow (\forall K \in \ZF \cup \{Z\}) K \subseteq \bar{\Theta}(F)
\Leftrightarrow (\ZF, Z) \tilde{\Theta} (\{F\}, F) 
\Leftrightarrow \{(\ZF, Z)\} \subseteq \bar{\tilde{\Theta}}(\{(\{F\}, F)\}).
\end{equation*}
Thus, the right-hand side in \eqref{eq-rh} implies that 
\begin{equation*}
(\exists \{(\ZF, Z)\} \in \tilde{\FF}) \{(\ZF, Z)\} \subseteq \bar{\tilde{\Theta}}(\{(\{F\}, F)\}) \wedge \{(\ZF, Z)\} \tilde{\Delta} \{(\XF', F')\},
\end{equation*}
form which we obtain with \ref{dn-apprel}\eqref{dn-apprel-2} that $\{\{F\}, F)\} \tilde{\Delta} \{(\XF', F')\}$, which means that for all $K' \in \XF' \cup \{F'\}$, $F \Delta K'$. By \ref{dn-apprel}\eqref{dn-apprel-5} there is hence $M' \in \tilde{\FF'}$ such that $\bigcup \XF' \cup \{F'\} \subseteq \bar{\Theta'}(M')$ and $F \Delta M'$. The first property implies that for all $K' \in \XF' \cup \{F'\}$, $K' \subseteq \bar{\Theta'}(M')$, that is, $M' \Upsilon_{\UU'} \{(\XF', F')\}$. Thus, we have 
\begin{equation*}
(\exists M' \in \tilde{\FF'}) F \Delta M' \wedge M' \Upsilon_{\UU'} \{(\XF', F')\},
\end{equation*}
which shows that $F (\Delta \circ \delta_{\UU'}) \{(\XF', F')\}$.

Now, conversely, assume that $F (\Delta \circ \delta_{\UU'}) \{(\XF', F')\}$. then there is some $M' \in \tilde{\FF'}$ with $F \Delta M'$ and $M' \Upsilon_{\UU'} \{(\XF', F')\}$, where the latter means that for all $K' \in \XF' \cup \{F'\}$, $K' \subseteq \bar{\Theta'}(M')$. With \ref{dn-apprel}\eqref{dn-apprel-3} we therefore obtain that for all $K' \in \XF' \cup \{G'\}$, $F \Delta K'$. So, $\{(\{F\}, F)\} \tilde{\Delta} \{(\XF', F')\}$. Because of \ref{dn-apprel}\eqref{dn-apprel-4} there is some $\{(\ZF, Z)\} \in \tilde{\FF}$ with 
\begin{equation*}
\{(\ZF, Z)\} \tilde{\Delta} \{(\XF', F')\} \quad\text{and}\quad \{(\ZF, Z)\} \subseteq \bar{\tilde{\Theta}}(\{(\{F\}, F)\}).
\end{equation*}
The latter implies that for all $K \in \ZF \cup \{Z\}$, $K \subseteq \bar{\Theta}(F)$, which means that $F \Upsilon_{\UU} \{(\ZF, Z)\}$. Thus, we have that there is some $\{(\ZF, Z)\} \in \tilde{\FF}$ so that $F \delta_{\UU} \{(\ZF, Z)\}$ and $\{(\ZF, Z)\} \tilde{\Delta} \{(\XF', F')\}$. Hence, $F (\delta_{\UU} \circ \tilde{\Delta}) \{(\XF', F')\}$. 
\end{proof}

Let us summarize what we have just shown.
\begin{proposition} \label{pn-deltaiso}
$\fun{\delta}{\ID_{\CFA}}{\EEE \circ \CCC}$ is a natural isomophism. 
\end{proposition}

Finally, we show that there is also a natural isomorphism $\fun{\gamma}{\ID_{\INF}}{\CCC \circ \EEE}$. Let to this end $\bA$ be an information frame. Then 
\begin{equation*}
\CCC(\EEE(\bA)) =  (\tilde{A}, (\tilde{\con}_{F})_{F \in \tilde{A}}, (\Vvdash_{F})_{F \in \tilde{A}})
\end{equation*}
 with
\begin{align*}
&\tilde{A} = \set{\{(X, i)\}}{i \in A \wedge X \in \con_{i}}, \\
&\tilde{\con}_{\{(X,i)\}} = \{\{(X, i)\}\} \cup \powf{\set{\{(Y, j)\} \in \tilde{A}}{X \vdash_{i} \{j\} \cup Y}}, \\
&\XF \Vvdash_{\{(X,i)\}} \{(Y, j)\} \Leftrightarrow (\exists (Z, e) \in \XF \cup \{(X,i)\}) Z \vdash_{e} \{j\} \cup Y.
\end{align*}

For $a, i \in A$, $X \in \con_{i}$, $\{(Y, j)\} \in \tilde{A}$ and $\XF \in \tilde{\con}_{\{(X,i\}}$ define
\begin{align*}
&X Q^{\bA}_{i} \{(Y, j)\} \Leftrightarrow X \vdash_{i} \{j\}  \cup Y, \\
&\XF P^{\bA}_{\{(X, i)\}} a \Leftrightarrow (\exists (Z, c) \in \XF \cup \{(X, i)\}) Z \vdash_{c} a.
\end{align*}
and set $Q_{\bA} = (Q^{\bA}_{i})_{i \in A}$ and $P_{\bA} = (P^{\bA}_{\{(X, a)\}})_{\{(X, i)\} \in \tilde{A}}$.

\begin{lemma}\label{lem-propPQ}
\begin{enumerate}
\item\label{lem-propPQ-1}
 $\appmap{Q_{\bA}}{\bA}{\CCC(\EEE(\bA))}$ such that truth elements are respected whenever $\bA$ has one. 
 
\item\label{lem-propPQ-2}
 $\appmap{P_{\bA}}{\CCC(\EEE(\bA))}{\bA}$ such that truth elements are respected whenever $\bA$ has one. 
 
 \item\label{lem-propPQ-3}
 $Q_{\bA} \circ P_{\bA} = \Id_{\bA}$.
 
 \item\label{lem-propPQ-4}
$P_{\bA} \circ Q_{\bA} = \Id_{\CCC(\EEE(\bA))}$.
\end{enumerate}
\end{lemma}
\begin{proof}
\eqref{lem-propPQ-1} We have to verify Conditions~\ref{dn-am}:

\ref{dn-am}\eqref{dn-am-1} Assume that $X Q^{\bA}_{i} (\{\{(Z, k)\}\} \cup \YF)$. Then we have for all $\{(Y, c)\} \in \YF$ that 
\begin{equation}\label{eq-propPQ1}
X Q^{\bA}_{i} \{(Y, c)\}, \quad\text{that is}\quad X \vdash_{i} (\{c\} \cup Y).
\end{equation}
Suppose in addition that $\YF \Vvdash_{\{(Z, k)\}} \{(V, e)\}$. Then there is some $\{(Y', c')\} \in \YF \cup \{(Z, k)\}$ with $Y' \vdash_{c'} \{e\} \cup V$. By \eqref{eq-propPQ1} it follows with \ref{dn-infofr}\eqref{dn-infofr-4} that $\{c'\} \in \con_{i}$ and  $Y \in \con_{i}$. Thus, $Y' \vdash_{i} \{e\} \cup V$. Since by \eqref{eq-propPQ1} also $X \vdash_{i} Y'$, it follows with \ref{dn-infofr}\eqref{dn-infofr-6} that $X \vdash_{i} \{e\} \cup V$, that is, $X Q^{\bA}_{i} \{(V, e)\}$.
 
\ref{dn-am}\eqref{dn-am-2}
Assume that $X Q^{\bA}_{i} \{(Y, b)\}$ and $X \subseteq X'$ with $X' \in \con_{i}$. Then we have that $X \vdash_{i} \{b\} \cup Y$; by weakening thus also $X' \vdash_{i} \{b\} \cup Y$, that is $X' Q^{\bA}_{i} \{(Y, b)\}$.

\ref{dn-am}\eqref{dn-am-3}
Suppose that $X \vdash_{i} X'$ and $X' Q^{\bA}_{i} \{(Y, b)\}$. Then we have that $X' \vdash_{i} \{b\} \cup Y$. It follows that  $X \vdash_{i}  \{b\} \cup Y$, that is, $X Q^{\bA}_{i} \{(Y, b)\}$.

\ref{dn-am}\eqref{dn-am-4}
Suppose that $\{i\} \in \con_{j}$ and $X Q^{\bA}_{i} \{(Y, b)\}$. Thus, $X \vdash_{i} \{b\} \cup Y$ and hence $X \vdash_{j} \{b\} \cup Y$, because of \ref{dn-infofr}\eqref{dn-infofr-9}. So, $X Q^{\bA}_{j }\{(Y, b)\}$.

\ref{dn-am}\eqref{dn-am-5}
Assume that $X Q^{\bA}_{i} \ZF$. Then we have that for all $\{(K, k)\} \in \ZF$ that $X Q^{\bA}_{i} \{(K, k)\}$, that is, $X \vdash_{i} \{k\} \cup K$. It follows that $X \vdash_{i} \bigcup \set{\{k\} \cup K}{\{(K, k)\} \in \ZF}$. By \ref{dn-infofr}\eqref{dn-infofr-11} there are thus $c, e \in A$, $U \in \con_{c}$ and $V \in \con_{e}$ so that $X \vdash_{i} \{c\} \cup U$, $U \vdash_{c} \{e\} \cup V$ and $V \vdash_{e} \bigcup \set{\{k\} \cup K}{\{(K, k)\} \in \ZF}$. Hence, we have that $X \vdash_{i} \{c\} \cup U$, $U Q^{\bA}_{c} \{(V, e)\}$ and $\{(V, e)\} \Vvdash_{\{(V,e)\}} \ZF$. Set $\VF = \{\{(V, e)\}\}$. Then $\VF \in \tilde{\con}_{\{(V, e)\}}$.

\ref{dn-am}\eqref{dn-am-6}
Assume that $\bA$ has a truth element $\bt$, Then $\CCC(\EEE(\bA))$ has truth element $\{(\emptyset, \bt)\}$, by Theorems~\ref{thm-inf-cfa}\eqref{thm-inf-cfa-2} and \ref{thm-cfa-inf}\eqref{thm-cfa-inf-2}. By Condition~\eqref{cn-T} we have that $\emptyset \vdash_{\bt} \{\bt\} \cup \emptyset$, that is $\emptyset Q^{\bA}_{\bt} \{(\emptyset, \bt)\}$.

\eqref{lem-propPQ-2} Again we have to verify Conditions~\ref{dn-am}:

\ref{dn-am}\eqref{dn-am-1}
Assume that $\XF P^{\bA}_{\{(X, a)\}} (\{k\} \cup Z)$ and $Z \vdash_{k} b$. Since $\XF \in \tilde{\con}_{\{(X, a)\}}$ we have that either $\XF = \{\{(X, a)\}\}$, or $\XF \neq \{\{(X, a)\}\}$ and for all $\{(V, c)\} \in \XF$, $X \vdash_{a} \{c\} \cup V$. 

If $\XF =  \{\{(X, a)\}\}$, it follows from the first assumption that $X \vdash_{a} \{k\} \cup Z$.  With \ref{dn-infofr}\eqref{dn-infofr-4} we then obtain in particular that $\{k\} \in \con_{a}$. Hence, it follows with the second assumption that  $Z \vdash_{a} b$. With the cut rule we finally have that $X \vdash_{a} b$, that is, $\XF P^{\bA}_{\{(X, a)\}} b$.

In the other case that $\XF \neq \{\{(X, a)\}\}$, we have that
\begin{equation}\label{eq-propPQ2}
(\forall \{(Y, d)\} \in \XF) X \vdash_{a} \{d\} \cup Y
\end{equation}
As $\XF P^{\bA}_{\{(X, a)\}} (\{k\} \cup Z)$, we moreover have that
\begin{equation}\label{eq-propPQ3}
(\forall e \in \{k\} \cup Z) (\exists \{(Y_{e}, d_{e})\} \in \XF \cup \{(X, a)\}) Y_{e} \vdash_{d_{e}} e,
\end{equation}
Set $I = \set{e \in \{k\} \cup Z}{(Y_{e}, d_{e}) = (X, a)}$ and $I' = (\{k\} \cup Z) \setminus I$. Then we have for all $e \in I'$ that $ X \vdash_{a} \{d_{e}\} \cup Y_{e}$. Because of \eqref{eq-propPQ2}, $\{d_{e}\} \in \con_{a}$. Hence, as a consequence of \eqref{eq-propPQ3},  $Y_{e} \vdash_{a} e$ and therefore $X \vdash_{a} e$. Thus, we have that $X \vdash_{a} I'$. Since for $e \in I$, $X \vdash_{a} e$, that is, $X \vdash_{a} I$. Overall we get that $X \vdash_{a} \{k\} \cup Z$. It follows that $\{k\} \in \con_{a}$. Thus, $Z \vdash_{a} b$ and hence $X \vdash_{a} b$, which shows that $\XF P^{\bA}_{\{(X, a)\}} b$.

\ref{dn-am}\eqref{dn-am-2}
Let $\XF, \XF' \in \tilde{\con}_{\{(X, a)\}}$ so that $\XF \subseteq \XF'$ and $\XF P^{\bA}_{\{(X, a)\}} b$. The latter assumption implies that there is some $\{(Z, k)\} \in \XF \cup \{\{(X, a)\}\}$ with $Z \vdash_{k} b$. As $\XF \subseteq \XF'$, $\{(Z, k)\} \in \XF' \cup \{\{(X, a)\}\}$ as well. So, also $\XF' P^{\bA}_{\{(X, a)\}} b$ holds.

\ref{dn-am}\eqref{dn-am-3}
Assume that $\XF \Vvdash_{\{(X, a\}} \XF'$ and $\XF' P^{\bA}_{\{(X, a)\}} b$. By the first assumption we have that for all $\{(Y, c)\} \in \XF'$ there is some $\{(Z_{\{(Y,c)\}}, k_{\{(Y, c)\}})\} \in \XF \cup \{\{(X, a)\}\}$ with $Z_{\{(Y, c)\}} \vdash_{k_{\{(Y, c)\}}} \{c\} \cup Y$. With the second assumption it follows that there is some $\{(\bar{Y}, \bar{c})\} \in \XF' \cup \{\{(X, a)\}\}$ with $\bar{Y} \vdash_{\bar{c}} b$. As $Z_{\{(\bar{Y}, \bar{c})\}} \vdash_{k_{\{(\bar{Y}, \bar{c})\}}} \{\bar{c}\} \cup \bar{Y}$, we have that $\{\bar{c}\} \in \con_{k_{\{(\bar{Y}, \bar{c})\}}}$. Therefore, $\bar{Y} \vdash_{k_{\{(\bar{Y}, \bar{c})\}}} b$. Since moreover, $Z_{\{(\bar{Y}, \bar{c})\}} \vdash_{k_{\{(\bar{Y}, \bar{c})\}}} \bar{Y}$, we have that $Z_{\{(\bar{Y}, \bar{c})\}} \vdash_{k_{\{(\bar{Y}, \bar{c})\}}} b$, that is, $\XF P^{\bA}_{\{(X, a)\}} b$.

\ref{dn-am}\eqref{dn-am-4}
Let $\{(Y, c)\} \in \tilde{\con}_{\{(X, a)\}}$ and $\XF P^{\bA}_{\{(Y, c)\}} b$.  Then $\tilde{\con}_{\{(Y,c)\}} \subseteq \tilde{\con}_{\{(X, a)\}}$ and therefore $\XF \in \tilde{\con}_{\{(X, a)\}}$. Moreover, there is some $\{(Z, k)\} \in \XF \cup \{\{(Y, c)\}\}$ with $Z \vdash_{k} b$. Without restriction let $(Y, c) \neq (X, a)$. Since $\{(Y, c)\} \in \tilde{\con}_{\{(X, a)\}}$, we have that $X \vdash_{a} \{c\} \cup Y$. If $(Z, k) = (Y, c)$, it follows that $X \vdash_{a} \{k\} \cup Z$. Hence, $\{k\} \in \con_{a}$ and thus $Z \in \con_{a}$ and $Z \vdash_{a} b$. Consequently, $X \vdash_{a} b$.
In case that $(Z, k) \neq (Y, c)$ we have that $\{(Z, k)\} \in \XF$. Thus, in both cases there is some $\{(Z', k')\} \in \XF \cup \{\{(X, a)\}\}$ with $Z' \vdash_{k'} b$. So, we have that $\XF P^{\bA}_{\{(X, a)\}} b$.

\ref{dn-am}\eqref{dn-am-5}
Assume that $\XF P^{\bA}_{\{(X, a)\}} K$. Then, for each $k \in K$, there is some $\{(Z_{k}, c_{k})\} \in \XF \cup \{\{(X, a)\}\}$ such that $Z_{k} \vdash_{c_{k}} k$. If $(Z_{k}, c_{k}) = (X, a)$ then $X \vdash_{a} k$. In case that $(Z_{k}, c_{k}) \neq (X, a)$, then $\{(Z_{k}, c_{k})\} \in \XF$ and thus $X \vdash_{a} \{c_{k}\} \cup Z_{k}$. It follows again that $\{c_{k}\} \in \con_{a}$  and therefore $Z_{k} \vdash_{a} k$. So, also in this case we have that $X \vdash_{a} k$. Hence, $X \vdash_{a} K$. By \ref{dn-am}\eqref{dn-am-5} it now ensures that there are $b, e \in A$, $Y \in \con_{b}$ and $U \in  \con_{e}$ such that $X \vdash_{a} \{b\} \cup Y$, $Y \vdash_{b} \{e\} \cup U$ and $U \vdash_{e} K$. From $X \vdash_{a} \{b\} \cup Y$ we obtain that $\{\{(X, a)\}\} \Vvdash_{\{(X, a)\}} \{(Y, b)\}$. Set $\YF = \{\{(Y, b)\}\}$. Then $\XF \Vvdash_{\{(X, a)\}} (\{\{(Y, b)\}\} \cup \YF)$. Because of $Y \vdash_{b} \{e\} \cup U$, we have that $\YF P^{\bA}_{\{(Y, b)\}} (\{e\} \cup U)$ and, as seen, we have $U \vdash_{e} K$.

\ref{dn-am}\eqref{dn-am-6}
follows as in the proof of Statement~\eqref{lem-propPQ-1}.

\eqref{lem-propPQ-3}
Let $a, b \in A$ and $X \in \con_{a}$. Then we have
\begin{align*}
X (Q_{\bA} \circ &P_{\bA})_{a} b \\
\Leftrightarrow\mbox{} &(\exists \{(Y, c)\} \in \tilde{A}) (\exists \YF \in \tilde{\con}_{\{(Y, c)\}}) X Q^{\bA}_{a} (\{\{(Y, c)\}\} \cup \YF) \wedge \YF P^{\bA}_{\{(Y, c)\}} b \\
\Leftrightarrow\mbox{}  &(\exists \{(Y, c)\} \in \tilde{A}) (\exists \YF \in \tilde{\con}_{\{(Y, c)\}}) X \vdash_{a} \bigcup \set{
\{d\} \cup Z}{\{(Z, d)\} \in (\{\{(Y, c)\}\} \cup \YF)} \wedge \mbox{} \\ 
&(\exists \{(V, e)\} \in \{\{(Y, c)\}\} \cup \YF) V \vdash_{e} b. 
\end{align*}
Since $\{(V, e)\} \in \{\{(Y, c)\}\} \cup \YF$, we  have in particular that $X \vdash_{a} \{e\} \cup V$, which implies that $\{e\} \in \con_{a}$. Thus, $V \vdash_{a} b$. So, the right-hand side of the last equivalence above implies that $X \vdash_{a} b$, that is $X \Id^{\bA}_{a} b$.

If, conversely, $X \Id^{\bA}_{a} b$, then with Lemma~\ref{pn-amint}\eqref{pn-amint-2} we obtain that there are $e \in A$ and $V \in \con_{e}$ so that $X \vdash_{a} \{e\} \cup V$ and $V \vdash_{e} b$. Set $\YF = \{\{(V, e)\}\}$. Then $\YF \in \tilde{\con}_{\{(V, e)\}}$ and $\{(V, e)\} \in \tilde{A}$ so that $\YF P^{\bA}_{\{(V, e)\}} b$ and $X Q^{\bA}_{a} \{(V, e)\}$. The latter implies in particular that $X Q^{\bA}_{a} (\{\{(V, e\}\} \cup \YF)$. This shows that $X (Q_{\bA} \circ P_{\bA})_{a} b$.

\eqref{lem-propPQ-4}
Let $\{(X, a)\}, \{(Y, b)\} \in \tilde{A}$ and $\XF \in \tilde{\con}_{\{(X, a)\}}$. Then
\begin{align*}
\XF (P_{\bA} &\circ Q_{\bA})_{\{(X, a)\}} \{(Y, b)\}  \\
\Leftrightarrow\mbox{} &(\exists e \in A) (\exists Z \in \con_{e}) \XF P^{\bA}_{\{(X, a)\}} (\{e\} \cup Z) \wedge Z Q^{\bA}_{e} \{(Y, b)\}  \\
\Leftrightarrow\mbox{} &(\exists e \in A) (\exists Z \in \con_{e}) (\forall d \in \{e\} \cup Z) (\exists \{(V_{d}, c_{d})\} \in \{\{(X, a)\}\} \cup \XF) V_{d} \vdash_{c_{d}} d \wedge \mbox{} \\
& Z \vdash_{e} \{b\} \cup Y.
\end{align*}
Let $I = \set{d \in \{e\} \cup Z}{(V_{d}, c_{d}) = (X, a)}$ and $I' = (\{e\} \cup Z) \setminus I$. Then we have for all $d \in I$ that $X \vdash_{a} d$. Moreover, we have for all $d \in I'$ that $\{(V_{d}, c_{d})\} \in \XF$, from which we obtain by the definition of $\tilde{\con}_{\{(X, a)\}}$ that $X \vdash_{a} \{c_{d}\} \cup V_{d}$. Therefore, $\{c_{d}\} \in \con_{a}$. Hence, $V_{d} \vdash_{a} d$ and again $X \vdash_{a} d$. So, we have in both cases that $X \vdash_{a} \{e\} \cup Z$, which implies that $\{e\} \in \con_{a}$. Thus, $Z \vdash_{a} \{b\} \cup Y$. Finally, it follows that $X \vdash_{a} \{b\} \cup Y$, that is, we have that $\XF \Vvdash_{\{(X, a\}} \{(Y, b)\}$, which shows that $\XF \Id_{\CCC(\EEE(\bA))} \{(Y, b)\}$. 

Conversely, if $\XF \Vvdash_{\{(X, a\}} \{(Y, b)\}$, then there exist $\{(Z, e)\} \in \{\{(X, a)\}\} \cup \XF$ with $Z \vdash_{e} \{b\} \cup Y$. Hence, by \ref{dn-infofr}\eqref{dn-infofr-11}, there are $k \in A$ and $V \in \con_{k}$ so that $Z \vdash_{e} \{k\} \cup V$ and $V \vdash_{k} \{b\} \cup Y$. It follows that $\XF P^{\bA}_{\{(X, a\}} (\{k\} \cup V)$ and $V Q^{\bA}_{k} \{(Y, b)\}$, that is $\XF (P_{\bA} \circ Q_{\bA})_{\{(X, a)\}} \{(Y, b)\}$.
 \end{proof}

Set $\gamma_{\bA} = Q_{\bA}$.

\begin{lemma}\label{lem-natrgam}
$\fun{\gamma}{\ID_{\INF}}{\CCC \circ \EEE}$ is a natural transformation.
\end{lemma}
\begin{proof}
Let $\bA$, $\bA'$ be information frames and $\appmap{H}{\bA}{\bA'}$ an approximable mapping.  Set $\tilde{H} = \CCC(\EEE(H))$. Then for $\{(X, a)\} \in \tilde{A}$, $\XF \in \tilde{\con}_{\{(X, a)\}}$ and $\{(Y, b)\} \in \tilde{A'}$,
\begin{equation*}
\XF \tilde{H}_{\{(X, a)\}} \{(Y, b)\} \Leftrightarrow (\exists \{(Z, d)\} \in \{\{(X, a)\}\} \cup \XF) Z H_{d} (\{b\} \cup Y).
\end{equation*}
We have to show that $\gamma_{\bA} \circ \tilde{H} = H \circ \gamma_{\bA'}$, that is, $Q_{\bA} \circ \tilde{H} = H \circ Q_{\bA'}$. Let to this end $a \in A$, $X \in \con_{a}$ and $\{(Y, b)\} \in \tilde{A'}$. Then
\begin{equation}\label{eq-PQgam}
\begin{aligned}
X (Q&_{\bA} \circ \tilde{H})_{a} \{(Y, b)\} \\
\Leftrightarrow\mbox{} &(\exists \{(V, k)\} \in \tilde{A}) (\exists \VF \in \tilde{\con}_{\{(V, k)\}}) X Q^{\bA}_{a} (\{\{(V, k)\}\} \cup \VF) \wedge \VF \tilde{H}_{\{(V, k)\}} \{(Y, b)\} \\
\Leftrightarrow\mbox{}  &(\exists \{(V, k)\} \in \tilde{A}) (\exists \VF \in \tilde{\con}_{\{(V, k)\}})  X \vdash_{a} \bigcup\set{\{c\} \cup E}{\{(E, c)\} \in \{\{(V, k)\}\} \cup \VF} \wedge \mbox{} \\
&(\exists \{(Z, d)\} \in \{\{(V, k)\}\} \cup \VF) Z H_{d} (\{b\} \cup Y) .
\end{aligned}
\end{equation}
It follows in particular that $X \vdash_{a} \{d\} \cup Z$, from which we obtain that $\{d\} \in \con_{a}$. So, $Z H_{a} (\{b\} \cup Y)$. Moreover, $X \vdash_{a} Z$ and therefore $X H_{a} (\{b\} \cup Y)$, by \ref{dn-am}\eqref{dn-am-3}. By  Lemma~\ref{pn-amint}\eqref{pn-amint-2} there are $e \in A'$ and $W \in \con'_{e}$ with $X H_{a} (\{e\} \cup W)$ and $W \vdash'_{e} \{b\} \cup Y$, which means we have that $X H_{a} (\{e\} \cup W)$ and $W Q^{\bA'}_{e} \{(Y, b)\}$, That is, $X (H \circ Q_{\bA'})_{a} \{(Y, b)\}$.

Conversely, if $X (H \circ Q_{\bA'})_{a} \{(Y, b)\}$, then we have for some $e \in A'$ and $W \in \con'_{e}$ that $X H_{a} (\{e\} \cup W)$ and $W Q^{\bA'}_{e} \{(Y, b)\}$, and thus that $X H_{a} (\{e\} \cup W)$ and $W \vdash'_{e} \{b\} \cup Y$, With \ref{dn-am}\eqref{dn-am-1} it follows that $ X H_{a} (\{b\} \cup Y)$. By applying Lemma~\ref{pn-amint}\eqref{pn-amint-2} we obtain that there is some $d\in A$ and $Z \in \con_{d}$ with $X \vdash_{a} \{d\} \cup Z$ and $Z H_{d} (\{b\} \cup Y)$. Set $\VF = \{\{(Z, d)\}\}$. Then  $\VF \in \tilde{\con}_{\{(Z, d)\}}$ and the condition in the last  right-hand side of \eqref{eq-PQgam} is satisfied. So, we have that $X (Q_{\bA} \circ \tilde{H})_{a} \{(Y, b)\}$.
\end{proof}

In summary we obtained the following result.

\begin{proposition}\label{pn-gamiso}
$\fun{\gamma}{\ID_{\INF}}{\CCC \circ \EEE}$ is a natural isomorphism.
\end{proposition}

As consequence of Propositions~\ref{pn-deltaiso} and \ref{pn-gamiso} we now obtain the second central result of this paper.

\begin{theorem}\label{thm-cfainf}
The category $\CFA$ of CF-approximation spaces and CF-approximable relations is equivalent to the category $\INF$ of information frames and approximable mappings.
\end{theorem}

Because of Lemmas~\ref{lem-funcC} and \ref{lem-funcIC} we obtain further equivalence results between important subcategories of $\CFA$ und $\INF$, respectively.

\begin{corollary}\label{cor-subcfainf}
The following categories are equivalent as well:
\begin{enumerate}
\item
The category $\CFA$ of CF-approximation spaces and CF-approximable relations and the full subcategory $\sINF$ of strong information frames.

\item
The full subcategory $\tCFA$ of topological CF-approximation spaces and the full subcategories $\aINF$ and $\asINF$ of algebraic information frames and algebraic strong information frames,  respectively.

\item The full subcategory $\CFAM$ of CF approximation spaces with property~\eqref{cn-M} and the subcategories $\INF_{\bt}$ and $\sINF_{\bt}$ of information frames and strong information frames, which respectively have truth elements and approximable mappings that respect truth elements.

\item
The full subcategory $\tCFAM$ of topological CF-approximation spaces with Property~\eqref{cn-M} and the subcategories $\aINF_{\bt}$ and $\asINF_{\bt}$ of algebraic information frames and algebraic strong information frames, which respectively have truth elements and approximable maps respecting truth elements.

\end{enumerate}
\end{corollary}

With Theorems~\ref{thm-eqIFDOM} and \ref{thm-cfainf} also Wu and Xu's central result follows.

\begin{corollary}\label{cor-wx}
The categories $\CFA$ of CA-approximation spaces and CF-approximable relations and the category $\DOM$ of domains and Scott continuous functions are equivalent.
\end{corollary}

But with Corollaries~\ref{cor-eqIFDOM} and \ref{cor-subcfainf} we obtain more.

\begin{corollary}\label{cor-refwx}
The following categories are equivalent:
\begin{enumerate}

\item
The full subcategories $\tCFA$ and $\aDOM$ of topological CF-approximation spaces and algebraic domains, respectively.

\item
The full subcategories $\CFAM$ and $\DOM_{\bot}$ of CF-approximation spaces with Property~\eqref{cn-M} and pointed domains, respectively.

\item
The full subcategories $\tCFAM$ and $\aDOM_{\bot}$ of topological CF-approximation spaces with Property~\eqref{cn-M} and pointed algebraic domains, respectively.

\end{enumerate}
\end{corollary}

\section*{Acknowledgments}
The author is grateful to the anonymous referees for a careful reading of the manuscript and their constructive comments.

\bibliographystyle{amsplain}

\end{document}